\pdfoutput=1
\documentclass[acmsmall]{ec22acm}

\AtBeginDocument{%
  \providecommand\BibTeX{{%
    \normalfont B\kern-0.5em{\scshape i\kern-0.25em b}\kern-0.8em\TeX}}}
    
\usepackage[conf={text link section}]{template/proof-at-the-end}

\usepackage{amsmath,amsfonts,bm}

\def\eqref#1{equation~\ref{#1}}

\def\1{\bm{1}}

\def\vone{{\bm{1}}}

\def\vv{{\bm{v}}}

\def\evv{{v}}

\def\mA{{\bm{A}}}
\def\mB{{\bm{B}}}
\def\mC{{\bm{C}}}
\def\mD{{\bm{D}}}
\def\mE{{\bm{E}}}

\def\mI{{\bm{I}}}

\def\mP{{\bm{P}}}

\def\mV{{\bm{V}}}

\def\mPi{{\bm{\Pi}}}

\DeclareMathAlphabet{\mathsfit}{\encodingdefault}{\sfdefault}{m}{sl}
\SetMathAlphabet{\mathsfit}{bold}{\encodingdefault}{\sfdefault}{bx}{n}

\def\gE{{\mathcal{E}}}

\def\gV{{\mathcal{V}}}

\newcommand{\E}{\mathbb{E}}

\newcommand{\R}{\mathbb{R}}

\DeclareMathOperator*{\argmax}{arg\,max}

\DeclareMathOperator{\Tr}{Tr}

\newtheorem{thm}{Theorem}[section]
\newtheorem*{thm*}{Theorem}

\newtheorem{corollary}[thm]{Corollary}
\newtheorem*{corollary*}{Corollary}

\newtheorem{lemma}[thm]{Lemma}
\newtheorem*{lemma*}{Lemma}

\copyrightyear{2022}
\acmYear{2022}
\setcopyright{acmcopyright}
\acmConference[EC '22] {Proceedings of the 23rd ACM Conference on Economics and Computation}{July 11--15, 2022}{Boulder, CO, USA.}
\acmBooktitle{Proceedings of the 23rd ACM Conference on Economics and Computation (EC '22), July 11--15, 2022, Boulder, CO, USA}
\acmPrice{15.00}
\acmISBN{978-1-4503-9150-4/22/07}
\acmDOI{10.1145/3490486.3538322}

\usepackage{etoolbox}
\makeatletter
\patchcmd{\maketitle}{\@copyrightpermission}{
  \begin{minipage}{0.2\columnwidth}
    \href{http://creativecommons.org/licenses/by/4.0/}{\includegraphics[width=0.90\textwidth]{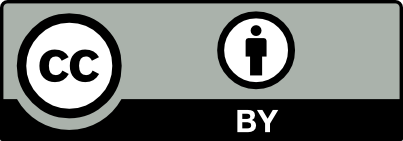}}
  \end{minipage}\hfill
  \begin{minipage}{0.8\columnwidth}
    \href{http://creativecommons.org/licenses/by/4.0/}{This work is licensed under a Creative Commons Attribution International 4.0 License.}
  \end{minipage}
 
  \vspace{5pt}
}{}{}

\makeatother

\acmSubmissionID{ecfp0022}

\begin{document}
\fancyhead{}

\title{On the Effect of Triadic Closure on Network Segregation}

\author{Rediet Abebe}
\email{rabebe@berkeley.edu}
\affiliation{
  \institution{University of California, Berkeley}
  \state{CA}
  \country{USA}
}

\author{Nicole Immorlica}
\email{nicimm@gmail.com}
\affiliation{
  \institution{Microsoft Research}
  \state{MA}
  \country{USA}
}

\author{Jon Kleinberg}
\email{kleinberg@cornell.edu}
\affiliation{
  \institution{Cornell University}
  \state{NY}
  \country{USA}
}

\author{Brendan Lucier}
\email{brlucier@microsoft.com}
\affiliation{
  \institution{Microsoft Research}
  \state{MA}
  \country{USA}
}

\author{Ali Shirali}
\email{shirali\_ali@berkeley.edu}
\orcid{0000-0003-3750-0159}
\affiliation{
  \institution{University of California, Berkeley}
  \state{CA}
  \country{USA}
}

\renewcommand{\shortauthors}{Abebe, Immorlica, Kleinberg, Lucier, and Shirali}

\begin{abstract}
The tendency for individuals to form social ties with others who are similar to themselves, known as homophily, is one of the most robust sociological principles. Since this phenomenon can lead to patterns of interactions that segregate people along different demographic dimensions, it can also lead to inequalities in access to information, resources, and opportunities. As we consider potential interventions that might alleviate the effects of segregation, we face the challenge that homophily constitutes a pervasive and organic force that is difficult to push back against. Designing effective interventions can therefore benefit from identifying counterbalancing social processes that might be harnessed to work in opposition to segregation. 

In this work, we show that triadic closure---another common phenomenon that posits that individuals with a mutual connection are more likely to be connected to one another---can be one such process. In doing so, we challenge a long-held belief that triadic closure and homophily work in tandem. By analyzing several fundamental network models using popular integration measures, we demonstrate the desegregating potential of triadic closure. We further empirically investigate this effect on real-world dynamic networks, surfacing observations that mirror our theoretical findings. We leverage these insights to discuss simple interventions that can help reduce segregation in settings that exhibit an interplay between triadic closure and homophily. We conclude with a discussion on qualitative implications for the design of interventions in settings where individuals arrive in an online fashion, and the designer can influence the initial set of connections. 
\end{abstract}



\maketitle

\section{Introduction}

Segregation impacts socioeconomic inequality by influencing individuals' abilities to obtain accurate and relevant information, garner social support, and improve access to opportunity~\cite{CalvoJackson,ccz,jackson2012social,banerjee2013diffusion,del2016echo,zeltzer2020gender,stoica2018algorithmic,twitter,taskrabbit}. A number of different social processes can impact segregation. Among these, homophily---the process by which individuals are more likely to form ties with whom they share similarities---is one of the most robust phenomena~\cite{lazarsfeld1954friendship,kossinets2009origins,mcpherson2001birds,mcpherson1987homophily,newman2002assortative,shrum1988friendship}. A long line of theoretical and empirical work shows that homophily can create and amplify existing segregation. And because homophily is a potent and organic force, it is challenging to push back against without harnessing existing social processes that may already be countering its negative effects.

In this work, we show that \emph{triadic closure}---a process in which individuals are more likely to form ties to others with whom they share mutual connections---is one such phenomenon ~\cite{rapoport1953spread,granovetter1977strength,kossinets2006empirical}. That is, we show that triadic closure alleviates segregation in settings where homophily is also present. Our results, which we present for a number of well-studied network formation models, challenge a long-held belief that triadic closure amplifies the effects of homophily. Such claims are frequently made, at times informally, citing concerns that homophily may lead friends-of-friends also to be similar, which would lead to further segregation under triadic closure \cite{asikainen,toth2019inequality,kossinets2009origins}. 
 
Our work challenges this intuition: Triadic closure connects people with mutual ties, and we may therefore assume that these new links reinforce existing patterns. We find, however, that the long-range nature of triadic closure can, in fact, counteract this phenomenon. In settings where homophily is present, individuals who are similar are more likely to form ties. Consequently, if friends-of-friends are not already connected, it may be because they are dissimilar. Triadic closure can therefore expose people to dissimilar individuals, thereby decreasing segregation. 

Mathematically, triadic closure operates on a graph-theoretic structure called a {\em wedge}. Wedges consist of two nodes that have a neighbor in common but are themselves not linked.  Triadic closure works by closing these wedges, i.e., by creating a link between these two nodes such that all three nodes are connected to one another. We analyze the effect of triadic closure on homophily by disaggregating wedges into monochromatic and bichromatic ones. The two nodes sharing a neighbor are of the same type in the case of the former but not the latter. We observe that the effect of triadic closure depends on the relative sizes of monochromatic and bichromatic wedges.  We study this effect both in an \emph{absolute} sense---by looking at whether network integration increases when we close a random wedge---and in a \emph{relative} sense---by comparing the effect of closing a random wedge with that of closing a random edge. 

We provide general results for a number of well-studied models, including the stochastic block model (SBM) and a popular growing network formation model by \citet{jr}, and show that triadic closure can have positive absolute and relative effects on integration in settings where there is homophily. We use these insights to study interventions on the Jackson-Rogers model and find that small changes leveraging the effects of triadic closure can have an outsized effect on mitigating segregation in the long run. We then study the interaction of homophily, triadic closure, and segregation using a large citation network where we estimate the network formation model and find that empirically observed effects of triadic closure on integration closely match our theoretical results. 

Our work also generalizes a number of theoretical contributions on graph and network theory. For instance, we generalize a result about network integration from \citet{jr} to a general network with heterogeneous nodes and with arbitrary distribution over the node types. There we provide general closed-form solutions for the time dynamic of network integration. Putting the relationship between triadic closure and homophily on a theoretical footing to ask these questions from a mathematical lens is a recent undertaking; in one formalization, \citet{asikainen} propose a model that combines triadic closure and random link rewiring with an underlying level of \emph{choice homophily}, in which nodes have a base preference for linking based on similarity. They show that the combination of these forces amplifies existing patterns of homophily. We examine these findings to show that this model introduces homophily even into the triadic closure process itself. We study a general variant of the \citet{asikainen} model and show that triadic closure mitigates segregation when all wedges are equally likely to close under triadic closure. 

The remainder of the paper is organized as follows:
In Section~\ref{sec:sbm}, we present an analysis of triadic closure in the stochastic block model, deriving mathematical
results on its absolute and relative effects on integration. We then introduce and analyze a growing graph model based on the Jackson-Rogers model, considering the effect of triadic closure on its equilibrium state integration in Section~\ref{sec:jr}. We then tackle the design of interventions that act on the initial phase of making friendships to optimize network integration. In Section~\ref{sec:asikainen}, we study a variant of the \citet{asikainen} model and show that in settings where triadic closure is not a priori biased in favor of monochromatic wedges, we obtain results consistent with our above findings. Finally, we study our results empirically using a large citation network and show that we can effectively model the network formation process in Section~\ref{sec:exp}. We also find that the effects of triadic closure on integration closely match our theoretical findings. We close with a discussion of related works as well as the interplay of homophily, triadic closure, integration, and implications for network interventions on- and off-line settings in Sections~\ref{sec:related_works}~and~\ref{sec:discussion}.

\section{Triadic Closure in The Stochastic Block Model}
\label{sec:sbm}

We begin by introducing notations and terminology which we will use throughout this paper: Let $G$ be a heterogeneous network, i.e., a network where nodes have a \emph{type}, which may, for instance, correspond to membership in a demographic group. We assume that there are $K$~types. We denote the type of node $i$ with $type(i)$. We say an edge~$(i,j)$ is \emph{monochromatic} if $type(i) = type(j)$ and \emph{bichromatic} otherwise.

Following convention, we measure network integration using the fraction of bichromatic edges. 
We denote the level of network integration at time~$t$ by~$f(t)$. Smaller values correspond to more-segregated networks.

A triplet of nodes~$(i, h, j)$ is called a \emph{wedge} if there exist edges~$(i, h)$ and $(h, j)$ but not $(i, j)$. A wedge is said to be monochromatic if $i$ and $j$ are of the same type and bichromatic if they are not. As is common in many studies of triadic closure, we assume that all wedges are equally likely to close under triadic closure. This is due to the fact that triadic closure is designed to capture the phenomena where the presence of node~$h$ in the wedge~$(i, h, j)$ impacts whether or not edge $(i, j)$ is eventually formed, regardless of the node types.

\vspace{2mm}
In this section, we study Stochastic Block Models (SBM). Under SBM, we assign independent probabilities to the existence of different edges, where these probabilities depend on the types of the corresponding nodes. Given nodes $i$ and $j$, edge $(i, j)$ is formed with probability~$p \in [0, 1]$ if $i$ and $j$ are of the same type and with probability $q \in [0, 1]$ if they are not. We say there is homophily if and only if $p > q$.

We study the effect of triadic closure in this model post network formation. That is, after the network is formed, we select and close a random wedge and measure the change in network integration. As is common in other studies on the influence of triadic closure, we first study the \textit{absolute effect} by comparing the state of network integration before and after the intervention. Our work also explores the \textit{relative effect} of triadic closure by considering an alternative mechanism as the baseline against which we compare the effect. We propose closing a random edge as this alternative mechanism and define relative effect as the difference in integration resulting from closing a random wedge versus a random edge.

\subsection{Absolute and Relative Effects of Triadic Closure}

We show that triadic closure improves network integration if and only if there is homophily.

\begin{theoremEnd}{thm}
\label{thm:sbm:abs_eff_on_integ}
For any SBM network $G$ with $K \ge 2$ types each consisting of $n_k$ nodes, where $k \in [K]$, for sufficiently large values of $n_k$, triadic closure has positive absolute effect on network integration if and only if $p > q$.
\end{theoremEnd}
\begin{proofEnd}
Let $e_b$ and $e_m$ ($w_b$ and $w_m$) be the expected number of bichromatic and monochromatic edges (wedges) respectively. The expected integration of the network before modification is approximately $\frac{e_b}{e_m+e_b}$.
\footnote{We approximated $\E[\frac{num}{den}]$ by $\frac{\E[num]}{\E[den]}$.}
After closing a random wedge, the expected (network) integration will be approximately $\frac{e_b + w_b/(w_m+w_b)}{e_m+e_b+1}$. For a fixed~$p$, we are looking for the range of~$q$ which increases the current level of integration:
\begin{align}
    \frac{e_b + w_b/(w_m+w_b)}{e_m+e_b+1} &\ge \frac{e_b}{e_m+e_b}   \nonumber \\
    \iff e_m e_b + e_b^2 + \frac{w_b}{w_m+w_b}(e_m + e_b) &\ge e_m e_b + e_b^2 + e_b  \nonumber \\
    \iff \frac{w_b}{w_m} &\ge \frac{e_b}{e_m}
    .
\end{align}

Let's define $n = \sum_{k \in [K]} n_k$ and $m_i = \sum_{k \in [K]} n_k^i$. By plugging $e_m$, $e_b$, $w_m$, and $w_b$ from Lemmas~\ref{lemma:n_edges_sbm}~and~\ref{lemma:n_wedges_sbm} into the above inequality:
\begin{align}
\label{eq:quad_ineq_general_sbm}
    \frac{w_b}{w_m} = \frac{(1 - q)q}{1 - p} \frac{2p(n m_2 - m_3) + q(n^3 + 2m_3 - 3nm_2)}{p^2 m_3 + q^2(nm_2 - m_3)} &\ge
    \frac{q}{p} \frac{n^2 - m_2}{m_2} = \frac{e_b}{e_m} \nonumber \\
    \iff q^2 \Big[-(1 - p)(n^2 - m_2)(n m_2 - m_3) - p m_2\big( n(n^2 - m_2) - 2(n m_2 - m_3) \big) \Big] &+ \nonumber \\
    q p \Big[-2(1 + p)m_2(nm_2 - m_3) + nm_2(n^2 - m_2) \Big] &+ \nonumber \\
    p^2 \Big[2m_2(nm_2-m_3) - (1-p)m_3(n^2 - m_2) \Big] &\ge 0
    .
\end{align}
Here we used the fact that $n m_2 \ge m_3$ from Lemma~\ref{lemma:n_m_ineq}. To determine the feasible region of $q$ we have to find the roots of LHS of Equation~\ref{eq:quad_ineq_general_sbm} which is a quadratic function in $q$. Let's define three new variables to simplify the equations: $A = n^2 - m_2$, $B = nm_2 - m_3$, and $C = m_3A - m_2B = n(nm_3 - m_2^2)$. According to Lemma~\ref{lemma:n_m_ineq}, $A$, $B$, and $C$ are all non-negative variables. The discriminant ($\Delta$) of the quadratic function can be found:
\begin{align}
\label{eq:delta}
    \frac{\Delta}{p^2} &= m_2^2 (-2(1 + p)B + nA)^2 + 4\big[ (1 - p)AB + pm_2(nA - 2B) \big]\big[ 2m_2 B - (1 - p)m_3 A \big] \nonumber \\
    &= m_2^2 (-2(1 - p)B + nA)^2 + 8m_2(1 - p)AB^2 - 4(1 - p)^2 m_3 A^2 B - p(1 - p)m_2 m_3(nA - 2B)A \nonumber \\
    &= (1 - p)^2 \big[ 4m_2^2 B^2 - 4m_3 A^2B + 4m_2 m_3 (nA - 2B)A \big] \nonumber \\
    &+ (1 - p) \big[ -4nm_2^2 A B + 8 m_2 A B^2 - 4m_2 m_3 A(nA - 2B) \big] \nonumber \\
    &+ n^2 m_2^2 A^2 \nonumber \\
    &= 4(1 - p)^2 C^2 - 4(1 - p)n m_2 A C + n^2 m_2^2 A^2 \nonumber \\
    &= \big(n m_2 A - 2(1 - p)C \big)^2
    .
\end{align}

As $\Delta > 0$, the quadratic has always two real roots:
\begin{align}
    q^*_1, q^*_2 &= \frac{p}{2} \frac{2(1 + p)m_2 B - n m_2 A \pm \sqrt{\Delta}}{-(1 - p)AB - p m_2 (n A - 2B)} \nonumber \\
    &= \frac{p}{2} \frac{2(1 - p)(-m_2B) + 4m_2 B - n m_2 A \pm (n m_2 A - 2(1 - p)C)}{(1 - p)(-AB + n m_2 A - 2 m_2 B) - m_2(n A - 2B)} \nonumber \\
    &= p \frac{(1 - p)(-m_2B - \pm C) + m_2 \big(2B - \frac{n A}{2} (1 - \pm 1) \big)}{(1 - p)(C - m_2 B) + m_2(2B - n A)} \nonumber \\
    &= p, \; p \frac{-(1 - p)m_3 A + 2 m_2 B}{(1 - p)(C - m_2 B) + m_2(2B - n A)}
    .
\end{align}
The first root is always $q^*_1 = p$ regardless of the network's structure. For the numerator of the second root, $q^*_2$, we have:
\begin{align}
\label{eq:num_q2}
    -(1 - p) m_3 A + 2 m_2 B &\ge 2 m_2 B - m_3 A \nonumber \\
    &= 2 m_2 (n m_2 - m_3) - m_3 (n^2 - m_2 ) \nonumber \\
    &= 2 n m_2^2 - n^2 m_3 - m_2 m_3 \ge 0
    ,
\end{align}
where we applied Lemma~\ref{lemma:n_m_ineq} to obtain the final inequality. Further, plugging $2 m_2 B \ge m_3 A$ into the denominator of $q^*_2$ gives:
\begin{align}
\label{eq:den_q2}
    (1 - p)(C - m_2 B)  + m_2(2B - n A) &= (1 - p)(m_3 A - 2m_2 B) + m_2(2B - n A) \nonumber \\
    &\le m_2(2B - n A) \nonumber \\
    &= m_2(2(n m_2 - m_3) - n(n^2 - m_2)) \nonumber \\
    &= m_2(n(m_2 - n^2) - 2 m_3) \nonumber \\
    &\le m_2 - n^2 \le 0
    ,
\end{align}
where we used Lemma~\ref{lemma:n_m_ineq} to obtain the last inequality. Note that the denominator of $q^*_2$ is exactly the coefficient of $q^2$ in the quadratic function of Equation~\ref{eq:quad_ineq_general_sbm}. So, Equation~\ref{eq:den_q2} says that this quadratic function is concave. On the other hand, Equations~\ref{eq:num_q2}~and~\ref{eq:den_q2} shows $q^*_2 \le 0$. So, we can conclude that the inequality of Equation~\ref{eq:quad_ineq_general_sbm} holds if and only if $q^*_2 \le 0 \le q \le q^*_1 = p$.

\end{proofEnd}

The proof first shows that closing a random wedge increases network integration if and only if the ratio of bichromatic wedges to monochromatic wedges is larger than the ratio of bichromatic edges to monochromatic edges. We then approximate the number of wedges and edges with their expected values and show that homophily is a necessary and sufficient condition to achieve the stated result. We note that this result holds for \emph{any} number of types as well as for cases where the types may be imbalanced in size, i.e., there may be a majority-minority partition.

Triadic closure may be improving integration simply because we are adding an edge and not because of the type of edge that was added. To untangle the effect of edge addition with that of triadic closure, we turn our attention to the relative effect.

\begin{theoremEnd}{thm}
\label{thm:sbm:rel_eff_on_integ}
Consider the baseline of adding a random edge to an $\text{SBM}$ network~$G$ with $K \ge 2$ types each consisting of $n_k$~nodes, where $k \in [K]$. For sufficiently large values of $n_k$:
\begin{enumerate}
    \item Triadic closure has a negative relative effect on network integration if $p > q$, 
    \item Triadic closure has positive or neutral relative effect on network integration if and only if $q \ge p \ge q \, l^*$, where 
    \begin{equation*}
        l^* = \frac{2(\sum_k n_k)(\sum_k n_k^2)^2 - (\sum_k n_k)^2(\sum_k n_k^3) - (\sum_k n_k^2)(\sum_k n_k^3)}{(\sum_k n_k^3)\big((\sum_k n_k)^2 - (\sum_k n_k^2)\big)} \le 1.
    \end{equation*}
    \end{enumerate}
\end{theoremEnd}
\begin{proofEnd}
Let $e_b$, $e_m$, $o_b$, $o_m$, $w_b$, and $w_m$ be the expected number of bichromatic edges, monochromatic edges, bichromatic missing edges, monochromatic missing edges, bichromatic wedges, and monochromatic wedges, respectively. After closing a random wedge, the expected integration will be approximately $\frac{e_b + w_b/(w_m+w_b)}{e_m+e_b+1}$. Similarly, after adding a random edge, the expected integration will be approximately $\frac{e_b + o_b/(o_m+o_b)}{e_m+e_b+1}$. For a fixed $q$, we are looking for the range of $p$ such that closing a random wedge increases network integration more than adding a random edge:
\begin{align}
    \frac{e_b + w_b/(w_m+w_b)}{e_m+e_b+1} &\ge \frac{e_b + o_b/(o_m+o_b)}{e_m+e_b+1}   \nonumber \\
    \iff \frac{w_b}{w_m+w_b} &\ge \frac{o_b}{o_m+o_b}  \nonumber \\
    \iff \frac{w_b}{w_m} &\ge \frac{o_b}{o_m}
    .
\end{align}

Let's define $n = \sum_{k \in [K]} n_k$ and $m_i = \sum_{k \in [K]} n_k^i$. By plugging $o_m$, $o_b$, $w_m$, and $w_b$ from Lemmas~\ref{lemma:n_edges_sbm}~and~\ref{lemma:n_wedges_sbm} into the above inequality:
\begin{align}
\label{eq:rel_quad_ineq_general_sbm}
    \frac{w_b}{w_m} = \frac{(1 - q)q}{1 - p} \frac{2p(n m_2 - m_3) + q(n^3 + 2m_3 - 3nm_2)}{p^2 m_3 + q^2(nm_2 - m_3)} &\ge
    \frac{(1 - q)}{(1 - p)} \frac{n^2 - m_2}{m_2} = \frac{o_b}{o_m} \nonumber \\
    \iff p^2 \big[-m_3(n^2 - m_2) \big] &+ \nonumber \\
    p q \big[2m_2(n m_2 - m_3)\big] &+ \nonumber \\
    q^2 \big[n m_2(n^2 - m_2) - 2m_2(n m_2 - m_3) - (n^2 - m_2)(n m_2 - m_3) \big] &\ge 0
    .
\end{align}
Here we used the fact that $n m_2 \ge m_3$ from Lemma~\ref{lemma:n_m_ineq}. To determine the feasible region of $p$ we have to find the roots of LHS of Equation~\ref{eq:rel_quad_ineq_general_sbm} which is a quadratic function in~$p$. Let's define three new variables to simplify the equations: $A = n^2 - m_2$, $B = nm_2 - m_3$, and $C = m_3A - m_2B = n(nm_3 - m_2^2)$. According to Lemma~\ref{lemma:n_m_ineq}, $A$, $B$, and $C$ are all non-negative variables. The discriminant ($\Delta$) of the quadratic function is:
\begin{align}
\label{eq:rel_delta}
    \frac{\Delta}{4q^2} &= (m_2 B)^2 + m_3 A (n m_2 A - 2m_2 B - A B) \nonumber \\
    &= m_2 B^2 - 2m_2 m_3 A B + n m_2 m_3 A^2 - m_3 A^2 B \nonumber \\
    &= m_2 B^2 - 2m_2 m_3 A B + m_3 A^2 (n m_2 - B) \nonumber \\
    &= (m_3 A - m_2 B)^2 = C^2
    .
\end{align}

As $\Delta \ge 0$, the quadratic has always two \footnote{Possibly degenerate.} real roots:
\begin{align}
    p^*_1, p^*_2 &= q \frac{-m_2 B \pm C}{-m_3 A} \nonumber \\
    &= q, \; q(\frac{2m_2 B}{m_3 A} - 1)
    .
\end{align}
The first root is always $p^*_1 = q$ regardless of the network's structure. Expanding the ratio $\frac{2m_2 B}{m_3 A}$ appeared in the second root and applying Lemma~\ref{lemma:n_m_ineq}, we can see  $1 \le \frac{2m_2 B}{m_3 A} \le 2$. So, defining $l^* = \frac{2m_2 B}{m_3 A} - 1$, we have $0 \le p^*_2 = q \, l^* \le 1$. On the other hand, the coefficient of the $p^2$ term in the quadratic function of Equation~\ref{eq:rel_quad_ineq_general_sbm} is always negative. Therefore, the quadratic function is concave and $p$ satisfies the inequality of Equation~\ref{eq:rel_quad_ineq_general_sbm} if and only if $q \ge p \ge q \, l^*$. As $p > q$ does not fall into this range, a direct result is that triadic closure has negative relative effect when network is homophilous.

\end{proofEnd}

\vspace{2mm}

For the case of balanced groups (i.e., when the $n_k$ are all equal), Theorem~\ref{thm:sbm:rel_eff_on_integ} simplifies to:

\begin{corollary}
For any $\text{SBM}$ network~$G$ with $K \ge 2$ balanced types each consisting of $n/K$~nodes, for sufficiently large values of $n/K$, triadic closure has a neutral relative effect if $p = q$ and negative relative effect if $p > q$.
\end{corollary}
\begin{proof}
This follows from Theorem~\ref{thm:sbm:rel_eff_on_integ} if we set $n_k = \frac{n}{K}$, which results in $l^* = 1$.
\end{proof}

These above results show that we can obtain diverging conclusions when we consider absolute versus relative effects of triadic closure. In doing so, they highlight the need for further precision in examining the interaction between triadic closure, homophily, and related social phenomena. Namely, to isolate the effect of social phenomena such as triadic closure, we may need to set appropriate baselines against which we are comparing their effect.

We considered adding a random edge as a natural baseline in our setting but also note that a random edge is likely to be bichromatic. Another baseline we may consider is adding a homophilous random edge, i.e., rather than adding a random edge, we favor monochromatic edges using a factor~$\gamma \ge 1$. Let $o_m$ and $o_b$ be the expected number of monochromatic and bichromatic missing edges. The expected increase in the number of bichromatic edges after adding a $\gamma$-homophilous edge is approximately $\frac{o_b}{o_b + \gamma o_m}$. Theorem~\ref{thm:sbm:gamma_rel_eff_on_integ} shows that, compared to adding a homophilous edge, triadic closure has a positive relative effect if the network is sufficiently heterophilous.

\begin{theoremEnd}{thm}
\label{thm:sbm:gamma_rel_eff_on_integ}
Consider the baseline of adding a $\gamma$-homophilous random edge to an SBM network $G$ with $K \ge 2$ types consisting of $n_k$ nodes, where $k \in [K]$. For sufficiently large values of $n_k$, triadic closure has positive or neutral relative effect on network integration if and only if $q \, u(\gamma) \ge p \ge q \, l(\gamma)$, where $u(\gamma) \ge 1$ and $l(\gamma) \le 1$. 
\end{theoremEnd}
\begin{proofEnd}
We use the same notation as the proof of Theorem~\ref{thm:sbm:rel_eff_on_integ}. After closing a random wedge, the expected integration will be approximately $\frac{e_b + w_b/(w_m+w_b)}{e_m+e_b+1}$. After adding a $\gamma$-homophilous random edge, the expected integration will be approximately $\frac{e_b + o_b/(\gamma o_m+o_b)}{e_m+e_b+1}$. For a fixed $q$, we are looking for the range of $p$ such that closing a random wedge increases network integration more than adding a homophilous random edge:
\begin{align}
    \frac{e_b + w_b/(w_m+w_b)}{e_m+e_b+1} &\ge \frac{e_b + o_b/(\gamma o_m+o_b)}{e_m+e_b+1}   \nonumber \\
    \iff \frac{w_b}{w_m+w_b} &\ge \frac{o_b}{\gamma o_m+o_b}  \nonumber \\
    \iff \frac{w_b}{w_m} &\ge \frac{1}{\gamma} \frac{o_b}{o_m}
    .
\end{align}

By plugging $o_m$, $o_b$, $w_m$, and $w_b$ from Lemmas~\ref{lemma:n_edges_sbm}~and~\ref{lemma:n_wedges_sbm} into the above inequality:
\begin{align}
\label{eq:gamma_rel_quad_ineq_general_sbm}
    \frac{w_b}{w_m} = \frac{(1 - q)q}{1 - p} \frac{2p(n m_2 - m_3) + q(n^3 + 2m_3 - 3nm_2)}{p^2 m_3 + q^2(nm_2 - m_3)} &\ge
    \frac{1}{\gamma}\frac{(1 - q)}{(1 - p)} \frac{n^2 - m_2}{m_2} = \frac{1}{\gamma} \frac{o_b}{o_m} \nonumber \\
    \iff p^2 \big[-m_3 A \big] 
    + p q \big[2 \gamma m_2 B \big] 
    + q^2 \big[\gamma n m_2 A - 2 \gamma m_2 B - A B \big] &\ge 0
    .
\end{align}
Here we used the fact that $n m_2 \ge m_3$ from Lemma~\ref{lemma:n_m_ineq}. To determine the feasible region of~$p$, we have to find the roots of LHS of Equation~\ref{eq:gamma_rel_quad_ineq_general_sbm} which is a quadratic function in~$p$. The discriminant ($\Delta$) of the quadratic function is:
\begin{align}
\label{eq:gamma_rel_delta}
    \frac{\Delta}{4q^2} &= (\gamma A B)^2 + m_3 A (\gamma n m_2 A - 2 \gamma m_2 B - A B) \nonumber \\
    &= (\gamma m_2 B)^2 - 2 \gamma m_2 m_3 A B + m_3 A^2 (\gamma n m_2 - B) \nonumber \\
    &= (\gamma m_2 B - m_3 A)^2 + m_3 A^2 (\gamma n m_2 - B - m_3) \nonumber \\
    &= (\gamma m_2 B - m_3 A)^2 + (\gamma - 1) n m_2 m_3 A^2
    .
\end{align}

As $\Delta \ge 0$, the quadratic has always two real roots:
\begin{align}
    p^*_1, p^*_2 &= q \frac{\gamma m_2 B \pm \sqrt{(\gamma m_2 B - m_3 A)^2 + (\gamma - 1) n m_2 m_3 A^2}}{m_3 A} \nonumber \\
    &= q \, u(\gamma), q \, l(\gamma)
    .
\end{align}

We can bound $u(\cdot)$ by
\begin{align}
    u(\gamma) &= \frac{\gamma m_2 B + \sqrt{(\gamma m_2 B - m_3 A)^2 + (\gamma - 1) n m_2 m_3 A^2}}{m_3 A} \nonumber \\
    &\ge \frac{\gamma m_2 B + |\gamma m_2 B - m_3 A|}{m_3 A}
    = \begin{cases}
    \frac{2 \gamma m_2 B}{m_3 A} - 1 \ge 1 & \gamma \ge \frac{m_3 A}{m_2 B} \\
    1 & o.w.
    \end{cases} \nonumber \\
    &\ge 1
    ,
\end{align}
where we used $m_3 A \ge m_2 B$ that can be shown by utilizing Lemma~\ref{lemma:n_m_ineq}. Similarly we can bound $l(\cdot)$ by
\begin{align}
    l(\gamma) &= \frac{\gamma m_2 B - \sqrt{(\gamma m_2 B - m_3 A)^2 + (\gamma - 1) n m_2 m_3 A^2}}{m_3 A} \nonumber \\
    &\le \frac{\gamma m_2 B - |\gamma m_2 B - m_3 A|}{m_3 A}
    = \begin{cases}
    1 & \gamma \ge \frac{m_3 A}{m_2 B} \\
    \frac{2 \gamma m_2 B}{m_3 A} - 1 \le 1 & o.w.
    \end{cases} \nonumber \\
    &\le 1
\end{align}

Finally, as the coefficient behind the $p^2$ term in Equation \ref{eq:gamma_rel_quad_ineq_general_sbm} is always negative, the quadratic function is concave. Therefore, the inequality holds if and only if $q \, u(\gamma) = p^*_1 \ge p \ge = p^*_2 = q \, l(\gamma)$.

\end{proofEnd}

Note that for given $p$ and $q$, solving for $u(\gamma^*)=\frac{p}{q}$ if $p \ge q$ or $l(\gamma^*)=\frac{p}{q}$ if $p < q$, provides an equivalence notion for the effect of triadic closure. In this case, the effect of closing a random wedge on network integration is the same as adding a $\gamma^*$-homophilous edge to the network.

For the special case of balanced groups, Figure~\ref{fig:p_vs_gamma(lambda1)} shows $u(\gamma)$ and $l(\gamma)$ for different values of $K$. In this figure, the shaded area corresponds to the values of $\frac{p}{q}$ such that $u(\gamma) \ge \frac{p}{q} \ge l(\gamma)$. This is the region where triadic closure has a positive relative effect compared to adding a $\gamma$-homophilous edge. We can see that as we increase $\gamma$, triadic closure will have a more-positive effect for larger values of $p$. Further, we can see that $u(\gamma)$ increases for larger values of $K$.

To see the effect of heterogeneous sizes in the groups, we consider the case where each group~$k$ has $n_k = n_1 \lambda^k$~members. So, the larger the $\lambda$, the more variance in size across groups. Figure~\ref{fig:p_vs_gamma(k2)} shows $u(\gamma)$ and $l(\gamma)$ for different $\lambda$~values when the number of groups is fixed. By increasing $\lambda$, we see that $u(\gamma)$ decreases, indicating that triadic closure has a less positive effect as the relative sizes between the groups increases.

\begin{figure}[t]
\begin{minipage}{0.49\linewidth}
    \centering
    \includegraphics[width=0.95\linewidth]{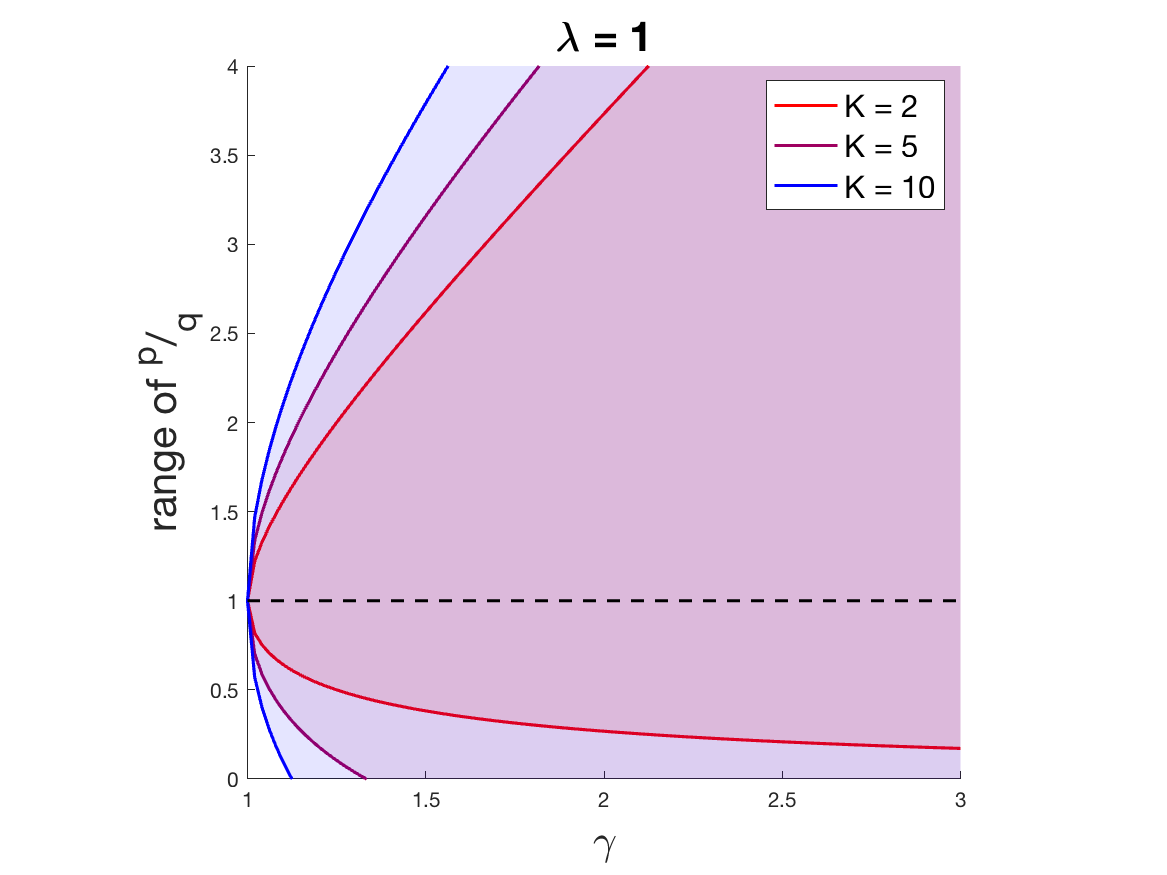}
    \caption{$u(\gamma)$ and $l(\gamma)$ for balanced groups.}
    \Description{For each $\gamma$, the region of $\frac{p}{q}$ such that triadic closure has a positive relative effect is plotted for different $K$s. The effect of triadic closure increases as $K$ increases.}
    \label{fig:p_vs_gamma(lambda1)}
\end{minipage}
\hfill
\begin{minipage}{0.49\linewidth}
    \centering
    \includegraphics[width=0.95\linewidth]{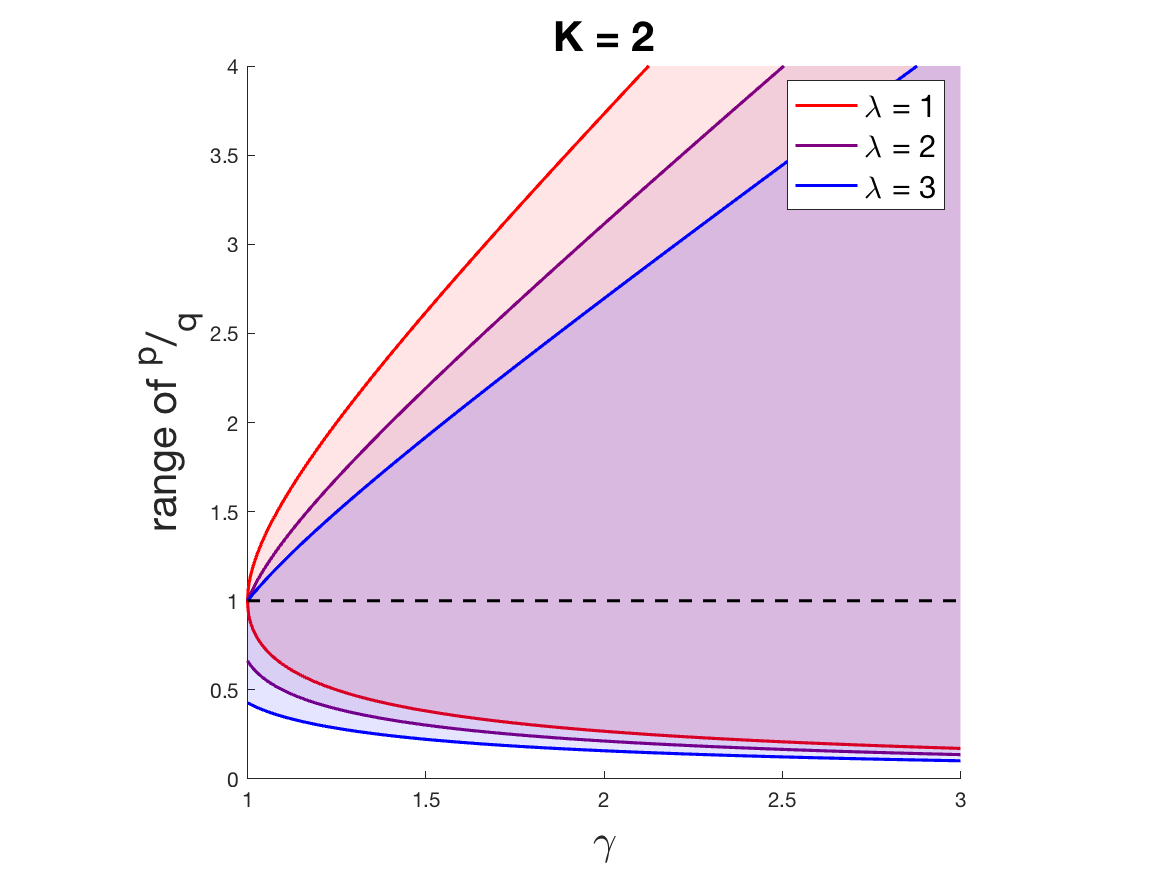}
    \caption{$u(\gamma)$ and $l(\gamma)$ for unbalanced groups ($\lambda=2$).}
    \Description{For each $\gamma$, the region of $\frac{p}{q}$ such that triadic closure has a positive relative effect is plotted for two groups but with different sizes. For more unbalanced groups, triadic closure is less effective.}
    \label{fig:p_vs_gamma(k2)}
\end{minipage}
\end{figure}

\subsection{Examining Other Measures of Network Health}

Thus far, we have studied integration using a popular measure in the literature---the fraction of bichromatic edges. High rates of network integration can be observed in settings where we may otherwise consider the network to be brittle. We therefore consider another robust measure of network health using eigenvector centrality. By doing so, we show that the positive effect of triadic closure is not limited to the original measure of integration.
 

Let $\mA$ be the network's adjacency matrix and $\vv(\mA)$ be the eigenvector corresponding to the largest eigenvalue of $\mA$. The eigenvector centrality of the $i^{th}$~node is defined as $\evv_i(\mA)$. Suppose we have a network consisting of two groups, including the setting where the groups may be imbalanced in size. Then our value of interest is the ratio of the average centrality of the minority to the average centrality of the majority group. As above, we first consider the absolute effect of triadic closure.

\begin{thm}
\label{thm:sbm:abs_eff_on_eigvec}
Consider an $\text{SBM}$ network $G$ with two types consisting of $n_1$ and $n_2$ nodes, where $n_1 > n_2$. Let $EV_k$ be the expected eigenvector centrality of a node from the $k^{th}$ group. For sufficiently large $n = n_1 + n_2$, triadic closure increases $\frac{EV_2}{EV_1}$ if and only if $p > q$.
\end{thm}

\begin{proof}
Let $P_{ij}$ be the probability that node~$i$ is connected to another node~$j$ in $G$. In an $SBM$, $P_{ij}=p$ if $type(i)=type(j)$ and $P_{ij}=q$ otherwise. We also set $P'_{ii}=0$ to avoid self loops. After closing a random wedge, we call the new network~$G'$ and the new probability that $i$ and $j$ are connected $P'_{ij}$. 

Let $w_{ij}$ be the expected number of wedges in $G$, such that we have edges $(i, h)$ and $(h, j)$ exist but not $(i, j)$. Let $w$ be the expected total number of wedges. With mean field approximation:
\begin{equation}
    P'_{ij} = P_{ij} + (1 - P_{ij}) \frac{w_{ij}}{w}
    .
\end{equation}
Here, the second term on the right hand side approximates the probability that $i$ and $j$ get connected after closing a random wedge. Note that $w=O(n^3)$ and $w_{ij}=O(n)$, so this term is $O \left( { \frac{1}{n^2} } \right)$. We find $w_{ij}$ based on $i$ and $j$'s types:
\begin{equation}
    w_{ij} = \begin{cases}
        (n_1 - 2)p^2 + n_2 q^2 & type(i)=type(j)=1 \\
        (n_2 - 2)p^2 + n_1 q^2 & type(i)=type(j)=2 \\
        (n - 2)p q & type(i) \neq type(j)
    \end{cases}
    .
\end{equation}

By plugging $w_{ij}$ into $P'_{ij}$, we note:
\begin{equation}
    P'_{ij} = \begin{cases}
        p'_1 = p + (1 - p)[(n_1 - 2)p^2 + n_2 q^2]\frac{1}{w} & type(i)=type(j)=1 \\
        p'_2 = p + (1 - p)[(n_2 - 2)p^2 + n_1 q^2]\frac{1}{w} & type(i)=type(j)=2 \\
        q' = q + (1 - q)(n - 2)p q \frac{1}{w} & type(i) \neq type(j)
    \end{cases}
    .
\end{equation}

Although we look for the expected eigenvector of the network, for a sufficiently large number of nodes, this quantity will be close to the eigenvector of the expected network~\cite{chung2011spectra, krishna}. We show the expected adjacency matrix of $G'$ by $\mA'=[P'_{ij}]$ and study eigenvectors of $\mA'$ instead of $G'$.

Due to the block nature of $\mA'$, it's easy to see the eigenvector corresponding to the largest eigenvalue of $\mA'$, which we denote by $\vv'$, has only two distinct values. Without loss of generality, we assume $\evv'_i = 1$ if $type(i)=1$ and $\evv'_i = a = \frac{EV_2}{EV_1}$ otherwise. That is, eigenvectors have a scale ambiguity that is usually resolved by setting the norm to one. Here, we instead fix element of the vector. Since $\mA' \vv' = \lambda \vv'$, we need to satisfy the following two equations: 
\begin{align}
    (n_1 - 1) p'_1 + n_2 q' a = \lambda \\
    n_1 q' + (n_2 - 1) p'_2 a = \lambda a
\end{align}
These give us a quadratic equation for $a$:
\begin{equation}
    a^2 [n_2 q'] + a [(n_1 - 1)p'_1 - (n_2 - 1)p'_2] - n_1 q' = 0
    .
\end{equation}
Dropping $O \left( { \frac{1}{n^3} } \right)$ from $p'_1$, $p'_2$, and $q'$ and plugging into the above equation, we get:
\begin{equation}
    a^2 \left[ { n_2 q(1 + (1-q) p \frac{n}{w}) } \right] 
    + a \left[ { (n_1 - n_2) p \left( { 1 + (1-p)p\frac{n}{w} } \right) } \right]
    - n_1 q \left( { 1 + (1-q) p \frac{n}{w} } \right) = 0.
\end{equation}
Defining $\beta = \sqrt{(n_1 - n_2)^2 p^2 + 4 n_1 n_2 q^2}$, the square root of the discriminant ($\Delta$) of this quadratic equation is:
\begin{equation}
    \sqrt{\Delta} = \beta \Big[
    1 + \frac{n}{w} \frac{(n_1 - n_2)^2 p^3 (1 - p) + 4 n_1 n_2 p q^2 (1 - q)}{\beta^2}
    \Big]
    .
\end{equation}
We can then find the solution corresponding to $a \ge 0$:
\begin{equation}
\label{eq:a_sbm_eigvec}
    a = \frac{\beta - (n_1 - n_2)p}{2n_2 q} 
    + \frac{n}{w}\frac{(n_1 - n_2)}{n_2}(p - q) p^2 \frac{\beta - (n_1 - n_2)p}{2 q \beta} + O \left( { \frac{1}{n^3} } \right)
    .
\end{equation}
This solution consists of two terms: The first term is exactly $\frac{EV_2}{EV_1}$ before closing a wedge. The second term is the change due to triadic closure. Since $\beta > (n_1 - n_2)p$, signs of $n_1 - n_2$ and $p - q$ determine the effect. Given group $1$ is the majority group, the effect of triadic closure on $a$ is positive if and only if $p > q$. 

\end{proof}

This above theorem shows that triadic closure can improve the centrality position of a minority group in an absolute sense. As above, we also examine this in a relative sense by comparing triadic closure with adding a $\gamma$-homophilous random edge.

\begin{theoremEnd}{thm}
\label{thm:sbm:gamma_rel_eff_on_eigvec}
Consider the baseline of adding a $\gamma$-homophilous random edge to an $\text{SBM}$ network $G$ with two types consisting of $n_1$ and $n_2$ nodes, where $n_1 > n_2$. Let $EV_k$ be the expected eigenvector centrality of a node from the $k^{th}$ group. For sufficiently large $n = n_1 + n_2$, triadic closure has positive relative effect on $\frac{EV_2}{EV_1}$ if and only if $\gamma > \frac{p}{q} c(p, q)$, where $c(p, q) \le 1$. Further, $c(p, q) > \frac{q}{p}$ if $p>q$.
\end{theoremEnd}

\begin{proofEnd}
We use a similar terminology as the proof of Theorem~\ref{thm:sbm:abs_eff_on_eigvec}. 

First, we find $P'_{ij}$ after closing a $\gamma$-homophilous edge:
\begin{equation}
    P'_{ij} = \begin{cases}
        p' = p + (1 - p) \frac{\gamma}{o_\gamma} & type(i) = type(j) \\
        q' = q + (1 - q) \frac{1}{o_\gamma} & type(i) \neq type(j)
    \end{cases}
    ,
\end{equation}
where $o_\gamma=o_b+\gamma o_m$, $o_m$ and $o_b$ are number of monochromatic and bichromatic missing edges respectively. Defining $a=\frac{EV_2}{EV_1}$, we obtain a quadratic equation
\begin{equation}
    a^2[n_2 q'] + a[n_1 - n_2]p' - n_1 q' = 0
    ,
\end{equation}
which has a single valid solution when we impose $a \ge 0$:
\begin{equation}
\label{eq:a_sbm_eigvec_rel}
    a = \frac{\beta - (n_1 - n_2)p}{2n_2 q} 
    + \frac{1}{o_\gamma}\frac{(n_1 - n_2)}{n_2}(1-p) \big[ \frac{p(1-q)}{q(1-p)} - \gamma \big] \frac{\beta - (n_1 - n_2)p}{2 q \beta} + O(\frac{1}{n^3})
    .
\end{equation}

A comparison of Equation~\ref{eq:a_sbm_eigvec_rel} and Equation~\ref{eq:a_sbm_eigvec} shows triadic closure has a positive relative effect if and only if
\begin{equation}
    \frac{n}{w} p^2 (p-q) > \frac{1}{o_\gamma} (1- p) \big[ \frac{p(1-q)}{q(1-p)} - \gamma \big].
\end{equation}
Simplifying this constraint, we have:
\begin{equation}
    \gamma > (\frac{p}{q}) \frac{w (1-q) - o_b n p q (p-q)}{w (1-p) + o_m n p^2 (p-q)} = \frac{p}{q} c(p, q)
    .
\end{equation}
Without going through details, we utilize Lemma~\ref{lemma:n_edges_sbm}, and expand $c(p, q)$:
\begin{equation}
    c(p, q) = \frac{(n_1^3 + n_2^3) p^2 + n n_1 n_2 q(2p + q)}
    {(n_1^3 + n_2^3) p^2 + n n_1 n_2 [p(2q + p) + \frac{1-p}{1-q}(q^2 - p^2)]}
    .
\end{equation}

Next we bound $c(p,q)$. We first show $c(p,q) \le 1$ by comparing the nominator and the denominator of $c(p, q)$:
\begin{align}
    q (2p + q) &\le p(2q + p) + \frac{1-p}{1-q}(q^2 - p^2) \nonumber \\
    \iff (p^2 - q^2)[\frac{1-p}{1-q} - 1] &\le 0 \nonumber \\
    \iff -(p - q)^2 &\le 0
    .
\end{align}

Finally, we show $\frac{p}{q} c(p, q) \ge 1$ if and only if $p > q$:
\begin{align}
    (n_1^3 + n_2^3) p^3 + n n_1 n_2 p q (2p + q) &\ge (n_1^3 + n_2^3) p^2 q + n n_1 n_2 [pq(2q + p) + \frac{1-p}{1-q}q(q^2 - p^2)]  \nonumber \\
    \iff 0 &\ge -(p - q) \Big[ (n_1^3 + n_2^3) p^2 + n n_1 n_2 p q + n n_1 n_2 q \frac{1-p}{1-q} \Big] \nonumber \\
    \iff p &> q
    .
\end{align}

\end{proofEnd}

The proof of Theorem~\ref{thm:sbm:gamma_rel_eff_on_eigvec} follows a similar process as that of Theorem~\ref{thm:sbm:abs_eff_on_eigvec}. The general idea is to approximate expected eigenvectors with eigenvectors of the expected network and then compare the change in the largest eigenvector due to adding an edge versus due to closing a wedge.

We saw in Theorem~\ref{thm:sbm:gamma_rel_eff_on_eigvec} that adding a random edge, which corresponds to $\gamma=1$, is a hard-to-beat baseline. In a homophilous network, $\frac{p}{q} c(p,q) > 1 = \gamma$, so the relative effect is always negative. However, we can also see from this theorem that compared to a more realistic alternative  ($\gamma > 1$), as long as the network is not very homophilous, i.e., $\frac{p}{q} < \gamma$, triadic closure exhibits a more favorable relative performance.

\section{Triadic Closure in the Jackson-Rogers Model}
\label{sec:jr}

The Jackson-Rogers model is an evolving model originally introduced for homogeneous networks~\cite{jr} and later extended to directed heterogeneous networks~\cite{bramoulle2012homophily}. Here, we use an extended version of the model, which gives us more control over the incorporation of triadic closure. 

The evolution of the network is defined over discrete time steps. At each step, a new node arrives and makes new connections in two phases. In the first phase, it randomly selects $N_S$ and $N_D$ initial friends from similar and dissimilar nodes, respectively. Note that edges are directed from the new node to the older ones. In the second phase, it chooses $N_F$ nodes from the set of nodes accessible through an outbound edge of an initial friend. Nodes already connected to the new node are excluded from this set. This process is also biased: $\alpha$ proportion of these $N_F$ nodes will be selected from the friends of the similar initial friends. The rest of the connections will be equally distributed towards the friends of the dissimilar initial friends. 

In the explained Jackson-Rogers model, $N_F$ exactly accounts for triadic closure, and we can directly control it to measure the effect while the network is evolving. This corresponds to the absolute effect. However, manipulating $N_F$ also changes the total number of new connections per node. To distinguish the effect of triadic closure from an increased number of edges, we adopt the notion of relative effect. We say triadic closure has a positive relative effect if increasing $N_F$, while $N = N_S + N_D + N_F$ and $\frac{N_S}{N_D}$ are kept fixed, results in a higher network integration.

We identify homophily in the first phase of the process by $N_S > N_D$. The definition of homophily in the second phase is not straightforward as it depends on the number of friends-of-friends of different types. Our analyses in the following sections are not sensitive to the selection of $\alpha$ as long as $0 < \alpha < 1$.

\subsection{Absolute and Relative Effects of Triadic Closure}

To study the expected behavior of an evolving network from the Jackson-Rogers model, we first prove the following theorem.

\begin{theoremEnd}{thm}
\label{thm:identity_1}
For an evolving Jackson-Rogers network~$G$ with $K$~types and parameters $N_S$, $N_D$, $N_F$, and $1 > \alpha > \frac{1}{K}$, the network integration converges to
\begin{equation}
    \frac{N_D + (1 - \alpha)N_F}{N_S + N_D + \frac{K}{K-1}(1 - \alpha) N_F}
\end{equation}
with the rate of $O \left( { t^{-\frac{N_S + N_D}{N}} } \right)$, regardless of the distribution of node types.
\end{theoremEnd}

\begin{proofEnd}
Let $\theta_t \in [K]$ shows the type of the node arrived at time~$t$ and $P(\theta_t = \theta) = p(\theta)$ independent of other nodes. We name the nodes based on their arrival time. For $t + 1 > t_0$, we define
\begin{equation}
    P_{t_0}^{t+1}(\theta_{t+1} | \theta_{t_0}) = P(\text{node } t+1 \text{ is of type } \theta_{t+1} \text{ and is connected to node } t_0)
    .
\end{equation}
Sometimes we represent $P_{t_0}^{t}(\cdot | \cdot)$ as matrix~$\mP_{t_0}^{t} \in \R^{K \times K}_+$. Let $n_t(\theta)$ be the total number of nodes of type $\theta$ until $t$, and $N = N_S + N_D + N_F$ be the total number of outbound edges from each node. For $\theta_{t + 1} = \theta_{t_0} = \theta$, the mean field approximation of $P_{t_0}^{t+1}(\theta | \theta)$ is
\begin{align}
\label{eq:p_theta_theta}
    P_{t_0}^{t+1}(\theta | \theta) &= p(\theta) \frac{N_S}{n_t(\theta)} \nonumber \\
    &+ p(\theta) \alpha N_F \frac{\sum_{\tau = t_0 + 1}^t P_{t_0}^\tau(\theta | \theta)}{n_t(\theta) N} \nonumber \\
    &+ p(\theta) (1 - \alpha) \frac{N_F}{K - 1} \sum_{\theta' \neq \theta} \frac{\sum_{\tau = t_0 + 1}^t P_{t_0}^\tau(\theta' | \theta)}{n_t(\theta') N}
    .
\end{align}
The above equation consists of three terms; in all of them $p(\theta)$ corresponds to the probability that node $t + 1$ is of type $\theta$. The first term of Equation~\ref{eq:p_theta_theta} shows the probability that node~$t + 1$ finds node~$t_0$ in the first phase: there are $n_t(\theta)$ nodes of type~$\theta$ and node $t + 1$ is going to select $N_S$ of them, so, the probability that node~$t_0$ is one of the $N_S$~selected~nodes is $\frac{N_S}{n_t(\theta)}$. The second term of Equation~\ref{eq:p_theta_theta} shows the probability that node~$t + 1$ finds node~$t_0$ through their same type friends: the expected number of edges going out from nodes of type~$\theta$ and entering node~$t_0$ is $\sum_{\tau = t_0 + 1}^t P_{t_0}^\tau(\theta | \theta)$. The total number edges exiting nodes of type~$\theta$ is $n_t(\theta) N$. So, the probability that node~$t + 1$ finds node~$t_0$ from the same type friends in the second phase is $\alpha N_F \frac{\sum_{\tau = t_0 + 1}^t P_{t_0}^\tau(\theta | \theta)}{n_t(\theta) N}$. Here we neglected that some edges exiting nodes of type~$\theta$ might end in nodes which are already connected to node~$t + 1$ in phase~$1$ because in long run these edges are negligible compared to $n_t(\theta) N$. The third term of Equation~\eqref{eq:p_theta_theta} corresponds to the probability that node~$t + 1$ finds node~$t_0$ through nodes of other types. Again $\sum_{\tau = t_0 + 1}^t P_{t_0}^\tau(\theta' | \theta)$ is the expected number of links from nodes of type $\theta' \neq \theta$ towards node~$t_0$ and $n_t(\theta') N$ is the total number of links exiting nodes of type $\theta'$. So, the probability that node~$t + 1$ finds node~$t_0$ in the second phase through nodes of type~$\theta'$ is $(1 - \alpha) \frac{N_F}{K - 1} \frac{\sum_{\tau = t_0 + 1}^t P_{t_0}^\tau(\theta' | \theta)}{n_t(\theta') N}$.

Similarly, the mean field approximation of $P_{t_0}^{t + 1}(\theta' | \theta)$, where $\theta' \neq \theta$, is
\begin{align}
\label{eq:p_thetap_theta}
    P_{t_0}^{t+1}(\theta' | \theta) &= p(\theta') \frac{N_D}{n_t(\theta)} \nonumber \\
    &+ p(\theta') \alpha N_F \frac{\sum_{\tau = t_0 + 1}^t P_{t_0}^\tau(\theta' | \theta)}{n_t(\theta') N} \nonumber \\
    &+ p(\theta') (1 - \alpha) \frac{N_F}{K - 1} \sum_{\theta'' \neq \theta'} \frac{\sum_{\tau = t_0 + 1}^t P_{t_0}^\tau(\theta'' | \theta)}{n_t(\theta'') N}
    .
\end{align}
The above equation also consists of three terms; in all of them $p(\theta')$ corresponds to the probability that node~$t + 1$ is of type~$\theta'$. The first term shows the probability that node~$t + 1$ finds node~$t_0$ in the first phase. The second and third terms are the probability that node~$t + 1$ finds node~$t_0$ in the second phase through nodes of type $\theta'$ and $\theta'' \neq \theta'$, respectively.

We approximate $n_t(\theta)$ in Equations~\ref{eq:p_theta_theta}~and~\ref{eq:p_thetap_theta} with its expected value: $n_t(\theta) = p(\theta) t$. To simplify the equations and to be consistent with original definitions and proofs of the Jackson-Rogers model~\cite{bramoulle2012homophily}, we define new variables: 
\begin{align}
\label{eq:def}
    m_r &\coloneqq \frac{N_S + N_D}{N} \nonumber \\
    m_s &\coloneqq \frac{N_F}{N} \nonumber \\
    [\mP]_{\theta', \theta} &\coloneqq p(\theta) 1_{\theta' = \theta} \nonumber \\
    [\mP_r]_{\theta', \theta} = P_r(\theta', \theta) &\coloneqq \begin{cases}
      \frac{N_S}{N_S + N_D} & \theta' = \theta \\
      \frac{1}{K-1} \frac{N_D}{N_S + N_D} & o.w.
    \end{cases} \nonumber \\
    [\mB_r]_{\theta', \theta} &\coloneqq [\mP \mP_r \mP^{-1}]_{\theta', \theta} = \frac{p(\theta') P_r(\theta', \theta)}{p(\theta)} \nonumber \\
    [\mP_s]_{\theta', \theta} = P_s(\theta', \theta) &\coloneqq \begin{cases}
      \alpha & \theta' = \theta \\
      \frac{1}{K-1} (1 - \alpha) & o.w.
    \end{cases} \nonumber \\
    [\mB_s]_{\theta', \theta} &\coloneqq [\mP \mP_s \mP^{-1}]_{\theta', \theta} = \frac{p(\theta') P_s(\theta', \theta)}{p(\theta)} \nonumber \\
    [\mPi_{t_0}^t]_{\theta', \theta} = \Pi_{t_0}^t(\theta' | \theta) &\coloneqq \sum_{\tau = t_0 + 1}^t P_{t_0}^\tau(\theta' | \theta)
    .
\end{align}

Now we can merge Equations~\ref{eq:p_theta_theta}~and~\ref{eq:p_thetap_theta} into a single equation with matrix operations:
\begin{equation}
\label{eq:diff_eq_Pi}
    \mP_{t_0}^{t+1} = \frac{N m_r}{t} \mB_r + \frac{m_s}{t} \mB_s \mPi_{t_0}^t 
    .
\end{equation}
The LHS of Equation~\ref{eq:diff_eq_Pi} can be approximated by
$\mP_{t_0}^{t+1} = \mPi_{t_0}^{t+1} - \mPi_{t_0}^t \approx \frac{\partial \mPi_{t_0}^t}{\partial t}$. Then Equation \ref{eq:diff_eq_Pi} will be a differential equation for $\mPi_{t_0}^t$. Given the initial condition $\mPi_{t_0}^{t_0} = 0$, the solution of this equation is
\begin{equation}
\label{eq:Pi_closed_form}
    \mPi_{t_0}^t = \frac{N m_r}{m_s} \big[(\frac{t}{t_0})^{m_s \mB_s} - \mI \big] \mB_s^{-1} \mB_r
    ,
\end{equation}
where by definition 
$(\frac{t}{t_0})^{m_s \mB_s} = \sum_{\mu=0}^\infty \frac{1}{\mu!} (m_s \ln(\frac{t}{t_0}))^\mu \mB_s^\mu$. $\mB_s$ is invertible iff $\mP_s$ is invertible. From Lemma~\ref{lemma:inv_p}, $\mP_s$ is also invertible iff $\alpha > \frac{1}{K}$. 

The variable~$\Pi_{t_0}^t(\theta' | \theta)$ shows the expected number of edges to node~$t_0$ from nodes of type~$\theta'$ arrived until~$t$, given type of node~$t_0$ is $\theta$. So, the expected number of monochromatic edges connected to node~$t_0$ at time~$t$ is:
\begin{equation}
\label{eq:mono_as_trace}
    \text{mono}_{t_0}^t = \sum_{\theta \in [K]} p(\theta) \Pi_{t_0}^t(\theta | \theta) = \Tr(\mP \mPi_{t_0}^t)
    .
\end{equation}
Plugging Equation~\ref{eq:Pi_closed_form} into Equation~\ref{eq:mono_as_trace}:
\begin{equation}
\label{eq:mono_expanded}
    \text{mono}_{t_0}^t = \frac{N m_r}{m_s} \sum_{\mu=1}^\infty \frac{1}{\mu !} (m_s \ln(\frac{t}{t_0}))^\mu \Tr(\mP \mB_s^{\mu - 1} \mB_r)
    .
\end{equation}
Here $\Tr(\mP \mB_s^{\mu - 1} \mB_r)$ can be simplified to $\Tr(\mP \mP_s^{\mu - 1} \mP_r)$ using definitions from Equation~\ref{eq:def}. According to Lemma~\ref{lemma:inv_p}, there are eigendecompositions for $\mP_r$ and $\mP_s$ as $\mV \mD_r \mV^{-1}$ and $\mV \mD_s \mV^{-1}$ (note that they have similar eigenvectors). Let's name $d_r = [\mD_r]_{i, i} = \frac{K N_S/(N_S + N_D) - 1}{K - 1}$ and $d_s = [\mD_s]_{i, i} = \frac{K\alpha - 1}{K - 1}$ for $i > 1$. With these representations and utilizing Lemma~\ref{lemma:inv_V} to calculate $\mV^{-1}$, we can further simplify
\begin{align}
\label{eq:trace_calc}
    \Tr(\mP \mP_s^{\mu - 1} \mP_r) &= \Tr(\mP \mV \mD_s^{\mu-1} \mD_r \mV^{-1}) \nonumber \\
    &= \Tr\Big(\mP \mV \big(d_r d_s^{\mu - 1} \mI  + (1 - d_r d_s^{\mu - 1}) \mE_{11} \big) \mV^{-1} \big) \nonumber \\
    &= d_r d_s^{\mu - 1} \Tr(\mP) + (1 - d_r d_s^{\mu - 1}) \Tr(\mP \mV \mE_{11} \mV^{-1}) \nonumber \\
    &= d_r d_s^{\mu - 1} + (1 - d_r d_s^{\mu - 1}) \Tr(\mP \frac{1}{K} \vone) \nonumber \\
    &= \frac{1}{K} + \frac{K - 1}{K} d_r d_s^{\mu - 1} 
    ,
\end{align}
where $\mE_{11}$ is a matrix with only one non-zero element $[\mE_{11}]_{1, 1} = 1$. We also used the fact that $\Tr(\mP) = 1$. Plugging this result into Equation~\ref{eq:mono_expanded}, we can find a simpler form for $\text{mono}_{t_0}^t$:
\begin{align}
\label{eq:mono_closed_form}
    \text{mono}_{t_0}^t &= \frac{N m_r}{m_s} \sum_{\mu=1}^\infty \frac{1}{\mu !} (m_s \ln(\frac{t}{t_0}))^\mu (\frac{1}{K} + \frac{K - 1}{K} d_r d_s^{\mu - 1}) \nonumber \\
    &= \frac{N m_r}{K m_s} \Big[ \sum_{\mu=0}^\infty (m_s \ln(\frac{t}{t_0}))^\mu - 1 + \sum_{\mu=0}^\infty (K - 1)\frac{d_r}{d_s}(m_s d_s \ln(\frac{t}{t_0}))^\mu - 1 \Big] \nonumber \\
    &= \frac{N m_r}{K m_s} \Big[ (\frac{t}{t_0})^{m_s} - 1 + (K - 1) \frac{d_r}{d_s} \big( (\frac{t}{t_0})^{m_s d_s} - 1 \big) \Big]
    .
\end{align}

Our Variable of interest is the total number of monochromatic edges until time $t$: $\text{mono}^t = \sum_{t_0 = 1}^t \text{mono}_{t_0}^t$. To calculate this sum, we first observe that for any $0 \le c < 1$, 
\begin{equation}
    \int_1^{t + 1} t_0^{-c} dt_0 < \sum_{t_0=1}^t t_0^{-c} < 1 + \int_1^t t_0^{-c} dt_0
    ,
\end{equation}
which gives $\sum_{t_0=1}^t t_0^{-c} = \frac{1}{1 - c} t^{1 - c} + O(1)$. Since $m_s < 1$ and $m_s d_s < 1$, we can use this approximation and find $\text{mono}^t$ from Equation~\ref{eq:mono_closed_form}:
\begin{align}
\label{eq:all_mono}
    \text{mono}^t &= \frac{N m_r}{K m_s} \Big[ \sum_{t_0=1}^t(\frac{t}{t_0})^{m_s} - t + (K - 1) \frac{d_r}{d_s} \big( \sum_{t_0=1}^t (\frac{t}{t_0})^{m_s d_s} - t \big) \Big] \nonumber \\
    &= \frac{N m_r}{K m_s} \Big[ \frac{t}{1 - m_s} - t + O(t^{m_s}) + (K - 1) \frac{d_r}{d_s} \big( \frac{t}{1 - m_s d_s} - t \big) + O(t^{m_s d_s}) \Big] \nonumber \\
    &= t \frac{N m_r}{K} \Big[ \frac{1}{1 - m_s} + \frac{(K - 1) d_r}{1 - m_s d_s} \Big] 
    + O(t^{m_s}) + O(t^{m_s d_s})
    .
\end{align}

Finally, we can find the integration in long run ($t \rightarrow \infty$):
\begin{align}
    f_\infty = \lim_{t \rightarrow \infty} f(t) &= 1 - \lim_{t \rightarrow \infty} \frac{\text{mono}^t}{N t} \nonumber \\
    &= 1 - \frac{m_r}{K} \Big[ \frac{1}{1 - m_s} + \frac{(K - 1) d_r}{1 - m_s d_s} \Big] \nonumber \\
    &= 1 - \frac{m_r}{K} \Big[ \frac{1 - m_s d_s + (K - 1)d_r(1 - m_s)}{(1 - m_s)(1 - m_s d_s)} \Big] \nonumber \\
    &= (\frac{K - 1}{K}) \frac{(1 - d_r)(N_S + N_D) + (1 - d_s)N_F}{N_S + N_D + (1 - d_s)N_F} \nonumber \\
    &= \frac{N_D + (1 - \alpha)N_F}{N_S + N_D + \frac{K}{K-1}(1 - \alpha) N_F}
    .
\end{align}
with the convergence rate~$O(t^{m_s - 1})=O(t^{-\frac{N_S + N_D}{N}})$ as seen from Equation~\ref{eq:all_mono}.

\end{proofEnd}

In the proof of Theorem~\ref{thm:identity_1}, we obtain a stronger result than the integration in equilibrium. Following \citet{bramoulle2012homophily}, we use a mean-field approximation to find a coupled differential equation of how the composition of neighbors of a node changes over time. We find a closed-form solution to this differential equation and aggregate the behavior of individual nodes to find network integration as a function of time. Understanding the dynamic of the network in time lets us study the effect of interventions in Section~\ref{sec:jr:interv}.

Theorem~\ref{thm:identity_1} enables us to study the effect of triadic closure on network integration in equilibrium. From this theorem, it is straightforward to see that in a network with homophily, increasing $N_F$, while $N_S$ and $N_D$ are unchanged, will increase network integration. We call this the absolute effect and formally state the observation in the following theorem. 

\begin{theoremEnd}{thm}
\label{thm:jr:abs_eff_on_integ}
For an evolving Jackson-Rogers network~$G$ with $K$~types and parameters~$N_S$, $N_D$, $N_F$, and $1 > \alpha > \frac{1}{K}$, triadic closure has a positive absolute effect on network integration if and only if $N_S > \frac{N_D}{K - 1}$. 
\end{theoremEnd}
\begin{proofEnd}
For fixed $N_S$, $N_D$, and $\alpha < 1$:
\begin{align}
    \frac{\partial f_\infty}{\partial N_F} &= \frac{(N_S + N_D + \frac{K}{K - 1}(1-\alpha)N_F)(1 - \alpha) - (N_D + (1 - \alpha)N_F)\frac{K}{K - 1}(1 - \alpha)}{(N_S + N_D + \frac{K}{K - 1}(1-\alpha)N_F)^2} \nonumber \\
    &= \frac{1 - \alpha}{(N_S + N_D + \frac{K}{K - 1}(1-\alpha)N_F)^2} (N_S - N_D \frac{1}{K - 1}) > 0
    .
\end{align}
So, increasing $N_F$ will increase the integration if and only if $N_S > \frac{N_D}{K-1}$.
\end{proofEnd}

As above, one might attribute the positive effect in Theorem~\ref{thm:jr:abs_eff_on_integ} to the increased number of connections per node. Next, we show that even when the total number of edges per node and the composition of neighbors in the first phase are kept fixed, i.e., $N=N_S+N_D+N_F$ and $\frac{N_S}{N_D}$ are maintained, increasing $N_F$ will improve network integration. 

\begin{theoremEnd}{thm}
\label{thm:jr:rel_eff_on_integ}
For an evolving Jackson-Rogers network~$G$ with~$K$ types and parameters $N_S$, $N_D$, $N_F$, and $1 > \alpha > \frac{1}{K}$, increasing $N_F$ subject to a fixed $N=N_S + N_D + N_F$ and $\frac{N_S}{N_D}$, results in a relative improvement in network integration if and only if $N_S > \frac{N_D}{K - 1}$.
\end{theoremEnd}
\begin{proofEnd}
To measure the relative effect, we assume the total number of edges a new node makes to be fixed and equals to $N = N_S + N_D + N_F$. Further, we assume $\frac{N_S}{N_D} = \gamma$. Homophily requires $\gamma > \frac{1}{K - 1}$.

Defining $\beta = \frac{N_F}{N}$, we first simplify integration from Theorem~\ref{thm:jr:abs_eff_on_integ}:
\begin{align}
    f_\infty &= 
    \frac{N_D + (1 - \alpha)N_F}{N - N_F + \frac{K}{K - 1}(1 - \alpha)N_F} \nonumber \\ 
    &= \frac{\frac{1 - \beta}{1 + \gamma} + (1 - \alpha) \beta} {(1 - \beta) + \frac{K}{K - 1}(1 - \alpha)\beta} \nonumber \\
    &= \frac{1}{\gamma + 1} 
    \frac{1 + \beta[(1 - \alpha)(1 + \gamma) - 1]} {1 + \beta (\frac{1}{K - 1} - \frac{K}{K - 1} \alpha)}
  .
\end{align}
Then we have
\begin{align}
    \frac{\partial f_\infty}{\partial \beta} &= 
    \frac{1}{\gamma + 1} 
    \frac{(1 + \beta \frac{1 - K \alpha}{K - 1})((1 - \alpha)(1 + \gamma) - 1) - (1 + \beta[(1 - \alpha)(1 + \gamma) - 1])(\frac{1 - K \alpha}{K - 1})}
    {(1 + \beta \frac{1 - K \alpha}{K - 1})^2} \nonumber \\
    &= \frac{1}{\gamma + 1} 
    \frac{(1 - \alpha)(1 + \gamma) - 1 - \frac{1 - K \alpha}{K - 1}}
    {(1 + \beta \frac{1 - K \alpha}{K - 1})^2} \nonumber \\
    &= (\frac{1 - \alpha}{\gamma + 1})
    \frac{\gamma - \frac{1}{K - 1}}
    {(1 + \beta \frac{1 - K \alpha}{K - 1})^2} > 0
    ,
\end{align}
which shows $\frac{\partial f_\infty}{\partial N_F} =\frac{1}{N} \frac{\partial f_\infty}{\partial \beta}$ is positive if and only if there is homophily.
\end{proofEnd}

In summary, Theorems~\ref{thm:jr:abs_eff_on_integ}~and~\ref{thm:jr:rel_eff_on_integ} show in a homophilous Jackson-Rogers evolving network, amplifying the role of triadic closure helps mitigate segregation. This effect is not due to making more connections, but rather due to the effect of triadic closure exposing nodes to dissimilar nodes.  

\subsection{Behavior Under a Series of Interventions}
\label{sec:jr:interv}

We study how interventions on a network evolving with the Jackson-Rogers model impact network integration in the short and long term. Here, we focus on interventions that act solely on the first phase. Recalling our motivating examples related to college dormitory assignments or recommendation of friendships when an individual joins an online platform, we note that an authority (i.e., university or platform, respectively) may have more leverage in this initial phase than subsequent steps which proceed through friend-of-friend searches. Such interventions that act as ``nudges'' in the initial phase have recently been popular in the fairness in recommender systems community; research in this space has explored the impact of bias in link formation or other selection on the long-term health of online platforms, with some work exploring the role of small nudges by the platform to mitigate inequalities or achieve other desirable social outcomes~\cite{ekstrand2016behaviorism,guy2015social,hutson2018debiasing,knijnenburg2016recommender,schnabel2018short,stoica2018algorithmic,su2016effect}.

In our analysis of interventions, we assume that the
number of links formed in the first phase is fixed. The
designer has the ability to change the proportion of mono versus bichromatic edges formed in the initial seeding phase subject to this sum constraint. This intervention imitates, for instance, dorm assignments where there is a fixed number of slots per dorm, but universities have the ability to change the composition of occupants in each dorm. We also consider the setting where the designer would like to optimize network integration subject to rate-of-change constraints on the network or on the time frame over which the intervention can occur. This is a model for scenarios where it may be costly, infeasible, or undesirable to introduce a dramatic change all at once.

In the following theorem, we first find out the extent interventions can change the network integration assuming the period of intervening is very shorter than the age of the network. 

\begin{theoremEnd}{thm}
\label{thm:jr:intervention}
Let $G(T)$ be an evolving Jackson-Rogers network at time $T$ with $K$ types and parameters $N_S$, $N_D$, $N_F$, and $1 > \alpha > \frac{1}{K}$. For each $i \in [I]$, we intervene on the first phase of the evolution by setting the number of similar and dissimilar initial friends to $N_S^{(T + i)}$ and $N_D^{(T + i)}$, respectively, while the total number of initial friends is kept fixed: $N_S^{(T + i)} + N_D^{(T + i)} = N_S + N_D$. Assuming $T >> I$: 

\begin{enumerate}
    \item At time $T + I$, the expected effect of $i^{th}$~intervention on network integration is approximately
    \begin{equation}
    \label{eq:interv_immediate}
        -\frac{1}{N(T + I)} \Big[1 + \frac{N_F}{N T}(I - i) \frac{K\alpha - 1}{K - 1} \Big] \Delta N_S^{(T + i)}
        ,
    \end{equation}
    where $\Delta N_S^{(T + i)} = N_S^{(T + i)} - N_S$.

    \item At time $t$ when a long time is passed from $T + I$ and the network is evolved with the original parameters~$N_S$~and~$N_D$ after the intervention period, the expected effect of $i^{th}$~intervention on network integration is approximately
    \begin{equation}
    \label{eq:interv_long}
        -\frac{1}{N} \big(\frac{t}{T}\big)^{\frac{N_F}{N} \frac{K\alpha - 1}{K - 1} - 1} \Delta N_S^{(T + i)}
        .
    \end{equation}
\end{enumerate}

\end{theoremEnd}
\begin{proofEnd}
We use the same notation as the proof of Theorem~\ref{thm:identity_1}. 

\noindent
\textit{Proof of part 1.} Applying Equation~\ref{eq:diff_eq_Pi} for the $i^{th}$ intervention, we have:
\begin{equation}
\label{eq:diff_eq_Pi_i_th_interv}
    \mP_{t_0}^{T + i} = \frac{N m_r}{T + i - 1} \mB_r^{(T + i)} + \frac{m_s}{T + i - 1} \mB_s \big(\mPi_{t_0}^T + \sum_{j=1}^{i-1} \mP_{t_0}^{T+j} \big)
    .
\end{equation}
Note that for $t \le t_0$, $\mP_{t_0}^{t}$ and $\mPi_{t_0}^{t}$ are zero by definition. Equation~\ref{eq:diff_eq_Pi_i_th_interv} is recursive and can be expanded:
\begin{align}
\label{eq:diff_eq_Pi_i_th_interv_expanded}
    \mP_{t_0}^{T + i} &= \frac{N m_r}{T + i - 1} \mB_r^{(T + i)} \nonumber \\
    &+ \frac{m_s}{T + i - 1} \mB_s \Big[ \mPi_{t_0}^T 
    + m_s \mB_s \mPi_{t_0}^T \sum_{j=1}^{i-1} \frac{1}{T+j-1} 
    + N m_r \sum_{j=1}^{i-1} \frac{\mB_r^{(T+j)}}{T+j-1}
    \Big] + O(\frac{1}{T^3})
    .
\end{align}

We use the notation~$\Delta (variable)$ to show the change of a $variable$ due to interventions, i.e., how a $variable$ is changed as $N_S^{(T+i)}$ deviates from $N_S$. Rewriting Equation~\ref{eq:diff_eq_Pi_i_th_interv_expanded} with this notation and dropping $O(\frac{1}{T^3})$:
\begin{equation}
    \Delta\mP_{t_0}^{T+i} = \frac{N m_r}{T + i - 1} \Delta\mB_r^{(T + i)}
    + \frac{N m_r m_s}{T^2} \mB_s \sum_{j=1}^{i-1} \Delta \mB_r^{(T+j)}
    .
\end{equation}
Here, we used the fact that $\Delta \mPi_{t_0}^T$ is zero as future interventions have no effect on previous connections. Now we can find the effect on $\mPi_{t_0}^{T+I}$. For~$t_0 \le T$, we have:
\begin{align}
\label{eq:delta_Pi}
    \Delta\mPi_{t_0}^{T+I} = \sum_{i=1}^I \Delta\mP_{t_0}^{T+i}
    &= N m_r \sum_{i=1}^I \frac{\Delta\mB_r^{(T+i)}}{T+i-1}
    + \frac{N m_r m_s}{T^2} \mB_s \sum_{i=1}^I \sum_{j=1}^{i-1} \Delta\mB_r^{(T+j)} \nonumber \\
    &= N m_r \sum_{i=1}^I \frac{\Delta\mB_r^{(T+i)}}{T+i-1}
    + \frac{N m_r m_s}{T^2} \mB_s \sum_{i=1}^{I-1} \Delta\mB_r^{(T+i)} (I - i) \nonumber \\
    &= N m_r \sum_{i=1}^I \Delta\mB_r^{(T+i)} \Big[ \frac{1}{T + i - 1} + \frac{m_s (I - i) }{T^2} \mB_s \Big]
    .
\end{align}
Note that similar equation can be found for~$t_0 > T$ if we start the sum from~$i=t_0 - T + 1$ in Equation~\ref{eq:delta_Pi}. Next, we can find the change in number of monochromatic edges connected to node~$t_0$ ($t_0 \le T$) until~$T+I$:
\begin{align}
    \Delta \text{mono}_{t_0}^{T+I} &= \Tr(\mP \Delta\mPi_{t_0}^{T+I}) \nonumber \\
    &= N m_r \sum_{i=1}^I \Big[ \frac{1 }{T + i - 1} \Tr(\mP \Delta\mB_r^{(T+i)}) + \frac{m_s (I - i)}{T^2} \Tr(\mP \Delta\mB_r^{(T+i)} \mB_s) \Big] \nonumber \\
    &= N m_r \sum_{i=1}^I \Big[\frac{1}{T + i - 1} \big(\frac{\Delta N_S^{(T+i)}}{N_S + N_D} \big) 
    + \frac{m_s (I - i)}{T^2} \big( \frac{\Delta N_S^{(T+i)}}{N_S + N_D} \frac{K \alpha - 1}{K - 1} \big) \Big] \nonumber \\
    &= \sum_{i=1}^I \Big[ \frac{1}{T+i-1} + \frac{m_s (I - i)}{T^2} \frac{K \alpha - 1}{K - 1} \Big] \Delta N_S^{(T+i)}
    .
\end{align}
To be concise, we didn't go through the calculation of $\Tr(\cdot)$ here but it follows a similar technique as used in the proof of Theorem~\ref{thm:identity_1}. Again, similar results can be obtained for $t_0 > T$ if we start the sum from $i=t_0 - T + 1$. Finally, we find the change in total number of monochromatic edges until~$T+I$
\begin{equation}
    \Delta \text{mono}^{T+I} = \sum_{t_0 = 1}^{T+I-1} \Delta \text{mono}_{t_0}^{T+I} = \sum_{i=1}^I \Big[ 1 + \frac{m_s (I-i)}{T} \frac{K \alpha - 1}{K - 1} \Big] \Delta N_S^{(T+i)}
    ,
\end{equation}
which gives the change of integration at time~$T + I$ as:
\begin{equation}
    \Delta f(T+I) = -\frac{\Delta \text{mono}^{T+I}}{N (T + I)} 
    = -\frac{1}{N (T + I)} \sum_{i=1}^I \Big[ 1 + \frac{m_s (I-i)}{T} \frac{K \alpha - 1}{K - 1} \Big] \Delta N_S^{(T+i)} 
    .
\end{equation}
This completes the first part of the proof.

\noindent
\textit{Proof of part 2.} Equation~\ref{eq:diff_eq_Pi} explains the network dynamic after interventions. Approximating $P_{t_0}^{t+1} = \mPi_{t_0}^{t+1} - \mPi_{t_0}^t$ by $\frac{\partial \mPi_{t_0}^t}{\partial t}$ gives a differential equation for $\mPi_{t_0}^t$ that has the general solution:
\begin{equation}
    \mPi_{t_0}^t = \frac{N m_r}{m_s} (\frac{t}{c_1})^{m_s \mB_s} \mB_s^{-1} \mB_r \mC_2
\end{equation}
where $c_1$ and $\mC_2$ should be determined by initial conditions.\footnote{Note that we could embed $c_1$ into $\mC_2$ but the current representation makes it easier for us to apply initial conditions.} Let $c_1 = T + I$, the initial condition requires:
\begin{equation}
    \frac{N m_r}{m_s} \mB_s^{-1} \mB_r \mC_2 = \mPi_{t_0}^{T+I}
    .
\end{equation}
The linear relationship of $\mC_2$ and $\mPi_{t_0}^{T+I}$ makes it easy to use our $\Delta$ notation:
\begin{equation}
    \Delta \mC_2 = \frac{m_s}{N m_r} \mB_r^{-1} \mB_s \Delta\mPi_{t_0}^{T+I}
    ,
\end{equation}
where $\mPi_{t_0}^{T+I}$ can be plugged in from Equation~\ref{eq:delta_Pi}. We can then find the change of $\mPi_{t_0}^t$ due to interventions:
\begin{equation}
    \Delta\mPi_{t_0}^t = (\frac{t}{T+I})^{m_s \mB_s} \Delta\mPi_{t_0}^{T+I}
    .
\end{equation}

The change in number of monochromatic edges connected to node~$t_0$ ($t_0 \le T$) is:
\begin{align}
    \Delta\text{mono}_{t_0}^t &= \Tr(\mP \Delta\Pi_{t_0}^t) \nonumber \\
    &= \sum_{\mu=0}^{\infty} \frac{1}{\mu!} (m_s \ln(\frac{t}{T+I}))^\mu \Tr\big( \mP \mB_s^\mu \Delta\mB_r^{(T+i)} \Delta\Pi_{t_0}^{T+I} \big) \nonumber \\
    &= \frac{N m_r}{T} \sum_{\mu=0}^{\infty} \frac{1}{\mu!} (m_s \ln(\frac{t}{T+I}))^\mu \sum_{i=1}^I \Tr(\mP \mB_s^\mu \Delta\mB_r^{(T+i)}) 
    + O(\frac{1}{T^2}) \nonumber \\
    &= \frac{N m_r}{T} \sum_{\mu=0}^{\infty} \frac{1}{\mu!} (m_s \ln(\frac{t}{T+I}))^\mu \sum_{i=1}^I \frac{K-1}{K} d_s^\mu \Delta d_r^{(T+i)}
    + O(\frac{1}{T^2}) \nonumber \\
    &= \frac{N m_r}{T} \frac{K-1}{K} \sum_{i=1}^I (\frac{t}{T+I})^{m_s d_s} \Delta d_r^{(T+i)}
    + O(\frac{1}{T^2}) \nonumber \\
    &= \frac{1}{T} \sum_{i=1}^I (\frac{t}{T+I})^{m_s d_s} \Delta N_S^{(T+i)}
    + O(\frac{1}{T^2}) 
    .
\end{align}
Here, we used a result from the proof of Theorem~\ref{thm:identity_1} to calculate the trace function (Equation~\ref{eq:trace_calc}). We should also note that similar results can be obtained for $t_0 > T$ by starting the sum from $i=t_0 - T + 1$. Now we can find the effect on the total number of monochromatic edges:
\begin{equation}
    \Delta\text{mono}^t = \sum_{t_0=1}^t \Delta\text{mono}_{t_0}^t
    = \sum_{i=1}^I (\frac{t}{T+I})^{m_s d_s} \Delta N_S^{(T+i)}
    + O(\frac{1}{T}) 
    .
\end{equation}
Finally, dropping $O(\frac{1}{T})$ from $\Delta\text{mono}^t$, the effect on integration is:
\begin{equation}
    \Delta f(t) = -\frac{\Delta \text{mono}^{t}}{N t} 
    = -\frac{1}{N} \sum_{i=1}^I (\frac{t}{T})^{m_s d_s} \Delta N_S^{(T+i)} 
\end{equation}
where $d_s=\frac{K\alpha - 1}{K - 1}$ and $m_s=\frac{N_F}{N}$.

\end{proofEnd}

Equation~\ref{eq:interv_immediate} shows two different ways that the $i^{th}$~intervention changes the network: the first term in the parenthesis corresponds to the direct impact on initial friends of the node~$T+i$, and the second term explains how future nodes amplify this initial effect through triadic closure. Important observations can be made from the first part of the Theorem~\ref{thm:jr:intervention} which are summarized in the following corollary.

\begin{corollary}
The immediate effect of an intervention on the network of Theorem~\ref{thm:jr:intervention} is
\begin{enumerate}
    \item Independent of other interventions,
    \item Negatively proportional to the change of $N_S$,
    \item Higher if the intervention is applied earlier,
\end{enumerate}
as long as the period that interventions are applied is very shorter than the age of the network ($T >> I$).
\end{corollary}
\begin{proof}
The first argument is obvious from Equation \ref{eq:interv_immediate}; the effect of $i^{th}$~intervention only depends on $i$. The second argument comes from the fact that $K\alpha - 1$ is always positive as $\alpha$ is assumed to be larger than $\frac{1}{K}$. So, the coefficient behind $\Delta N_S$ in Equation~\ref{eq:interv_immediate} is always negative. Finally, the effect of the $i^{th}$ intervention varies with $I - i$, so the older an intervention, the larger its effect. 
\end{proof}

The immediate effect of interventions (Equation~\ref{eq:interv_immediate}) might look in contrast to the long-term effect (Equation~\ref{eq:interv_long}). In fact, reducing $N_S$ has a positive impact on the number of bichromatic edges, which is sublinear in time. However, the number of total new edges is also increasing linearly over time, and integration is the ratio of these two numbers: $\frac{sublinear(t)}{t + T}$. As we assumed the network was old enough ($T >> 1$), in the short term, the relative change of the total number of edges is small, and the effect is driven by $\approx \frac{sublinear(t)}{T}$. However, in the long term, the change in the total number of edges is not negligible, and network integration follows $\approx \frac{sublinear(t)}{t}$.

Now that we can predict the expected effect an intervention has on the network, we can design optimum interventions to maximize network integration. However, there are always some constraints, e.g., the stability of the network, that limit the change a network can tolerate. We model all of these constraints as a limit on the rate of the change. The next theorem shows there is a greedy solution for optimum interventions subject to this constraint. 

\begin{thm}
The optimum interventions of Theorem~\ref{thm:jr:intervention} such that
\begin{align*}
    &\max_{\{N_S^{(T+j)}\}_{j\in[I]}}  f(T + I) \\
    &\; s.t. \;\; \forall j \in [I-1]: f(T + j + 1) - f(T + j) \le \Delta
    ,
\end{align*}
where $f(t)$ is network integration at time $t$, can be found greedily from:
\begin{equation}
    \Delta N_S^{(T + j)} = \max \Big\{
    -N_S, 
    -N T \Delta \, \big(\frac{T}{T-j}\big) + \Big[ 1 - \frac{N_F}{N} \frac{K\alpha - 1}{K-1} \Big] \frac{1}{T-j} \sum_{i=1}^{j-1} \Delta N_S^{(T + i)}  
    \Big\}
    .
\end{equation}
If $N_S \ge N T \Delta (\frac{T}{T - 2 I})$, there is a closed-form solution for optimum interventions during $j \in [I]$:
\begin{equation}
\label{eq:optimum_interv}
    \Delta N_S^{(T + j)} = -N T \Delta \Big( 1 + \frac{1}{T} - \frac{N_F}{N T} \frac{K\alpha - 1}{K - 1} \Big)^{j - 1}
    .
\end{equation}
These interventions achieve $f(T + I) - f(T) = I \, \Delta$.
\end{thm}

\begin{proof}
Equation~\ref{eq:interv_immediate} can be expanded to the first order of~$\frac{1}{T}$ as:
\begin{equation}
\label{eq:interv_immediate_first_order}
    -\frac{1}{N T} \Big[1 + \frac{N_F}{N T}(I - i) \frac{K\alpha - 1}{K - 1} - \frac{I}{T} \Big] \Delta N_S^{(T + i)}
    .
\end{equation}
The change of integration from time $T+j-1$ to $T+j$ due to an intervention at time $T+i$ ($i < j$) is
\begin{equation}
    \frac{1}{N T^2} \Big[ \frac{N_F}{N} \frac{K \alpha - 1}{K - 1} - 1 \Big] (-\Delta N_S^{(T+i)})
    ,
\end{equation}
where we simply found the difference of Equation~\ref{eq:interv_immediate_first_order} for $I=j$ and $I=j-1$. Now we can rewrite the rate of the change constraint from time $T+j-1$ to $T+j$ as:
\begin{equation}
\label{eq:constraint_delta_j}
    \frac{1}{N T}[1 - \frac{j}{T}] (-\Delta N_S^{(T+j)}) + 
    \sum_{i=1}^{j-1} \frac{1}{N T^2} \Big[ \frac{N_F}{N} \frac{K \alpha - 1}{K - 1} - 1 \Big] (-\Delta N_S^{(T+i)})
    \le \Delta
    .
\end{equation}
This is a linear constraint in terms of $\{\Delta N_S^{(T+i)}\}_i$. The objective function is also linear:
\begin{equation}
    f(T+I) = \frac{1}{N T} \sum_{j=1}^I \Big[1 + \frac{N_F}{N T}(I - j) \frac{K\alpha - 1}{K - 1} - \frac{I}{T} \Big] (-\Delta N_S^{(T + j)}) = \sum_{j=1}^I c_j (-\Delta N_S^{(T + j)})
    .
\end{equation}
Here $c_j$ is positive and decreasing in $j$. Let $-\Delta N_S^{(T+j)}=x_j$ ($j\in[I]$) be the optimum solution of the problem. We argue that for any $j\in[I]$, $x_j$ is
\begin{equation}
    \min \Big\{N_S,
    N T \Delta (\frac{T}{T-j}) + \Big[ 1 - \frac{N_F}{N} \frac{K\alpha - 1}{K-1} \Big] \frac{1}{T-j} \sum_{i=1}^{j-1} x_i \Big\}
    .
\end{equation}
Otherwise, we could increase $x_j$ to make the constraint of Equation~\ref{eq:constraint_delta_j} binding. This increase does not violate other constraints, since $\frac{N_F}{N} \frac{K\alpha - 1}{K-1} - 1 < 0$.

Now if $N_S \ge N T \Delta (\frac{T}{T - 2 I})$, we have
\begin{equation}
    N T \Delta (\frac{T}{T-j}) + \Big[ 1 - \frac{N_F}{N} \frac{K\alpha - 1}{K-1} \Big] \frac{1}{T-j} \sum_{i=1}^{j-1} x_i 
    \le N T \Delta (\frac{T}{T - I}) + N_S \frac{I}{T - I} \le N_S.
\end{equation}
So, $x_j \le N_S$ is never binding and 
$x_j \approx N T \Delta + \Big[ 1 - \frac{N_F}{N} \frac{K\alpha - 1}{K-1} \Big] \frac{1}{T} \sum_{i=1}^{j-1} x_i$ for all $j \in [I]$. This is a recursive equation for $x_j$. Let's define $y_j = \sum_{i=1}^j x_j$. The recursive definition for $y_j$ will be: 
\begin{equation}
    y_j - y_{j-1} = N T \Delta 
    + \frac{1}{T} \Big[ 1 - \frac{N_F}{N} \frac{K\alpha - 1}{K-1} \Big]  y_{j-1}
\end{equation}
and $y_1 = x_1 =  N T \Delta$. Taking $Z$-Transform from this recursive equation gives
\begin{equation}
    Y(z) = \frac{N T \Delta z^{-1}}{(1 - z^{-1})(1 - (1 + \frac{1}{T} - \frac{N_F}{N T} \frac{K \alpha - 1}{K - 1})z^{-1})}.
\end{equation}
By taking $Z^{-1}$-transform of $Y(Z)$ one can see
\begin{equation}
    y_j = \frac{N T \Delta}{\frac{1}{T} - \frac{N_F}{N T} \frac{K \alpha  - 1}{K - 1}}
    \Big[(1 + \frac{1}{T} - \frac{N_F}{N T} \frac{K \alpha  - 1}{K - 1})^j - 1 \Big], \;\; j \ge 1
\end{equation}
and Equation~\ref{eq:optimum_interv} can be obtained by $\Delta N_S ^{(T+j)} = -x_j = y_{j-1} - y_j$.

\end{proof}

\section{Triadic Closure in A Fixed-Node Evolving Model}
\label{sec:asikainen}

\citet{asikainen} propose a model with a fixed number of nodes and edges where the network evolves through random edge addition and triadic closure. The authors argue that triadic closure increases observed homophily relative to homophilous random link formation, i.e., triadic closure has a negative relative effect.

Here, we show that this result is specific to their definition of triadic closure which favors monochromatic wedges. In contrast, triadic closure is often studied in settings where wedges do not exhibit such a bias \cite{easley_kleinberg_2010}. Empirical work on real-world networks also supports this unbiased wedge closing assumption \cite{kossinets2006empirical}. We therefore study a variant of the \citet{asikainen} model where triadic closure does not differentiate between monochromatic and bichromatic wedges. 

We first present the model: Consider a network with a random initial structure and where nodes belong to one of two groups. At each iteration, a \textit{focal node} is selected uniformly at random. Then a \textit{candidate node} is chosen by triadic closure with probability~$c$ or uniformly at random with probability $1-c$. The parameter~$c$ controls the relative impact of triadic closure in the evolution of the network. Let $\theta$ be the focal node type and $\theta'$ be the candidate node type. A link is formed between focal and candidate nodes with probability $S_{\theta, \theta'}'$ if the candidate is selected by triadic closure and $S_{\theta, \theta'}$ otherwise. Following \citet{asikainen}, we assume $S_{\theta, \theta'} = s$ and $S_{\theta, \theta'}' = s'$ if $\theta = \theta'$, and $S_{\theta, \theta'} = 1 - s$ and $S_{\theta, \theta'}' = 1 - s'$ otherwise. To keep the number of edges constant while network is evolving, a random edge connected to the focal node is removed whenever it forms a new edge with a candidate node. 

In the original model of \citet{asikainen}, $s' = s$ and homophily is imposed by setting $s > \frac{1}{2}$. Following the definitions above, we argue setting $s' = s$ adds extra homophily to triadic closure. Instead, to be consistent with our definition of triadic closure, we set $s'=\frac{1}{2}$. That is, we do not distinguish between monochromatic and bichromatic wedges. The result below shows how this change to an unbiased triadic closure setting leads to results consistent with observations in the SBM and Jackson-Rogers models. 

\begin{theoremEnd}{thm}
\label{thm:asikainen}
For a fixed-node evolving network $G$ with two equiprobable types and parameters~$s$ and $s'=\frac{1}{2}$, triadic closure has a positive relative effect on network integration if and only if $1 > s > \frac{1}{2}$, compared to a random link formation. 
\end{theoremEnd}
\begin{proofEnd}
Let $\Theta = \{1, 2\}$ be the set of possible types of nodes and  $n_\theta$ shows the relative size of the group with type $\theta \in \Theta$. 
We define the following variables:
\begin{itemize}
    \item $(T^k)_{\theta'|\theta}$: starting from a node of type $\theta$ and going to neighbor nodes, the probability of ending in a node of type $\theta'$ in $k$ steps.
    \item $P_{\theta', \theta}$: the probability that a randomly selected edge is between nodes of types $\theta'$ and $\theta$ ($P_{\theta', \theta} = P_{\theta, \theta'}$ by definition). 
    \item $M_{\theta'|\theta}$: the probability that in an iteration a new edge is formed between nodes of types $\theta'$ and $\theta$ given that the focal node is of type $\theta$.
\end{itemize}
The mean field approximation for $M_{\theta'|\theta}$ is:
\begin{equation}
\label{eq:asikainen_m}
    M_{\theta'|\theta} = (1 - c) n_{\theta'} S_{\theta', \theta} + c (T^2)_{\theta'|\theta} S'_{\theta', \theta}
    ,
\end{equation}
where the first term corresponds to the case that the candidate node is found uniformly at random and the second term corresponds to triadic closure. Let's define one step in time equivalent to $L$~iterations where $L$ is the total number of links. The mean field approximation for time derivative of $P_{\theta, \theta}$ is:
\begin{equation}
\label{eq:asikainen_p}
    \frac{dP_{\theta, \theta}}{dt} = n_\theta M_{\theta|\theta} - n_\theta (\sum_{\theta' \in \Theta} M_{\theta'|\theta}) T_{\theta|\theta}
    .
\end{equation}
Here, the first term is the probability that the focal node is of type~$\theta$ and has created a link to a similar candidate node. The second term is the probability that after successfully creating a new edge, an edge between two nodes of type~$\theta$ is removed. As we only have two types of nodes in our setup, we use $\bar{\theta}$ to show the other type: $\{\bar{\theta}\} = \Theta \setminus \{\theta\}$. By plugging Equation~\ref{eq:asikainen_m} into Equation~\ref{eq:asikainen_p}:
\begin{align}
\label{eq:asikainen_p_expanded}
    \frac{dP_{\theta, \theta}}{dt} =&
    n_\theta \big[ (1 - c) n_\theta s + c (T^2)_{\theta|\theta} \big] \nonumber \\
    &- n_\theta T_{\theta|\theta} \big[(1-c) n_\theta s + c (T^2)_{\theta|\theta} s' + (1 - c)n_{\bar{\theta}} (1 - s) + c (T^2)_{\bar{\theta}|\theta}(1 - s')\big]
    .
\end{align}
The $(T^2)$ terms can be expanded recursively:
\begin{equation}
    (T^2)_{\theta'|\theta} = \begin{cases}
        T^2_{\theta|\theta} + T_{\bar{\theta}|\theta} T_{\theta|\bar{\theta}} = T^2_{\theta|\theta} + (1 - T_{\theta|\theta})(1 - T_{\bar{\theta}|\bar{\theta}}) & \theta'=\theta \\
        T_{\theta|\theta} T_{\bar{\theta}|\theta} + 
        T_{\bar{\theta}|\theta} T_{\bar{\theta}|\bar{\theta}} = T_{\theta|\theta}(1 - T_{\theta|\theta}) + (1 - T_{\theta|\theta}) T_{\bar{\theta}|\bar{\theta}} & \theta'=\bar{\theta}
    \end{cases}
    .
\end{equation}
We used $T_{\bar{\theta}|\theta} = 1 - T_{\theta|\theta}$ in obtaining the above equation. By plugging this expansion into Equation~\ref{eq:asikainen_p_expanded}, $dP_{\theta, \theta}/{dt}$ will be a function of $T_{\theta|\theta}$ and $T_{\bar{\theta}|\bar{\theta}}$ only. 

Although we cannot directly solve these differential equations, we can look into equilibrium that happens at the fixed points of equations. In fixed points, $dP_{\theta, \theta}/{dt} = 0$ for all $\theta \in \Theta$:
\begin{align}
    \frac{dP_{1, 1}}{dt}(T_{1|1}, T_{2|2}) = 0 \\
    \frac{dP_{2, 2}}{dt}(T_{1|1}, T_{2|2}) = 0
    .
\end{align}

By solving these coupled equations for $T_{1|1}$ and $T_{2|2}$, we find all fixed points. Then based on the second derivatives, we keep only the stable ones (look at \citet{asikainen} for details on distinguishing stable fixed points). 

Finally, there is a one-to-one relationship between $T_{1|1}, T_{2|2}$ and $P_{1, 1}, P_{2, 2}$:
\begin{equation}
    T_{\theta|\theta} = \frac{2 P_{\theta, \theta}}{2 P_{\theta, \theta} + P_{\theta, \bar{\theta}}} 
    = \frac{2 P_{\theta, \theta}}{1 + P_{\theta, \theta} - P_{\bar{\theta}, \bar{\theta}}}, \;\; \theta \in \Theta
\end{equation}
that can be solved to find $P_{1, 1}$ and $P_{2, 2}$ in terms of $T_{1|1}$ and $T_{2|2}$. The network integration is simply $1 - P_{1, 1} - P_{2, 2}$. Figure~\ref{fig:asikainen_integ_vs_s} shows network integration at stable fixed points for $n_1 = n_2 = \frac{1}{2}$ and $s'=\frac{1}{2}$, while varying $s$. One can see that increasing $c$ consistently increases (decreases) integration when $s > \frac{1}{2}$ ($s < \frac{1}{2}$).

\end{proofEnd}

Note that the condition $1 > s > 1/2$ corresponds to the setting where random link formation is homophilous. To better understand the extent to which Theorem~\ref{thm:asikainen} applies, we depict network integration theoretically estimated at equilibrium in Figure~\ref{fig:asikainen_integ_vs_s}. We have also marked simulated results with crosses to show that the theory and empirical observations closely match one another. We note that as we increase the impact of triadic closure by increasing $c$, integration increases if $s > \frac{1}{2}$ and decreases if $s < \frac{1}{2}$. There are two extreme cases to observe: In the case of no triadic closure ($c=0$), integration falls linearly with respect to $s$. On the other extreme, when edges only form via triadic closure, i.e., $c=1$, there are two possibilities: if groups of different types are initially completely segregated, the integration will always be zero regardless of $s$. If the network is not completely segregated, the resulting integration will be $0.5$ as there was no homophily. Another interesting observation is that even when the network is maximally homophilous ($s=1$), for large enough $c$, triadic closure will not let the integration go to zero. In sum, our above result in Theorem \ref{thm:asikainen} and corresponding simulations show that triadic closure works against segregation in homophilous networks. 

\begin{figure}[t]
    \centering
    \includegraphics[width=0.49\linewidth]{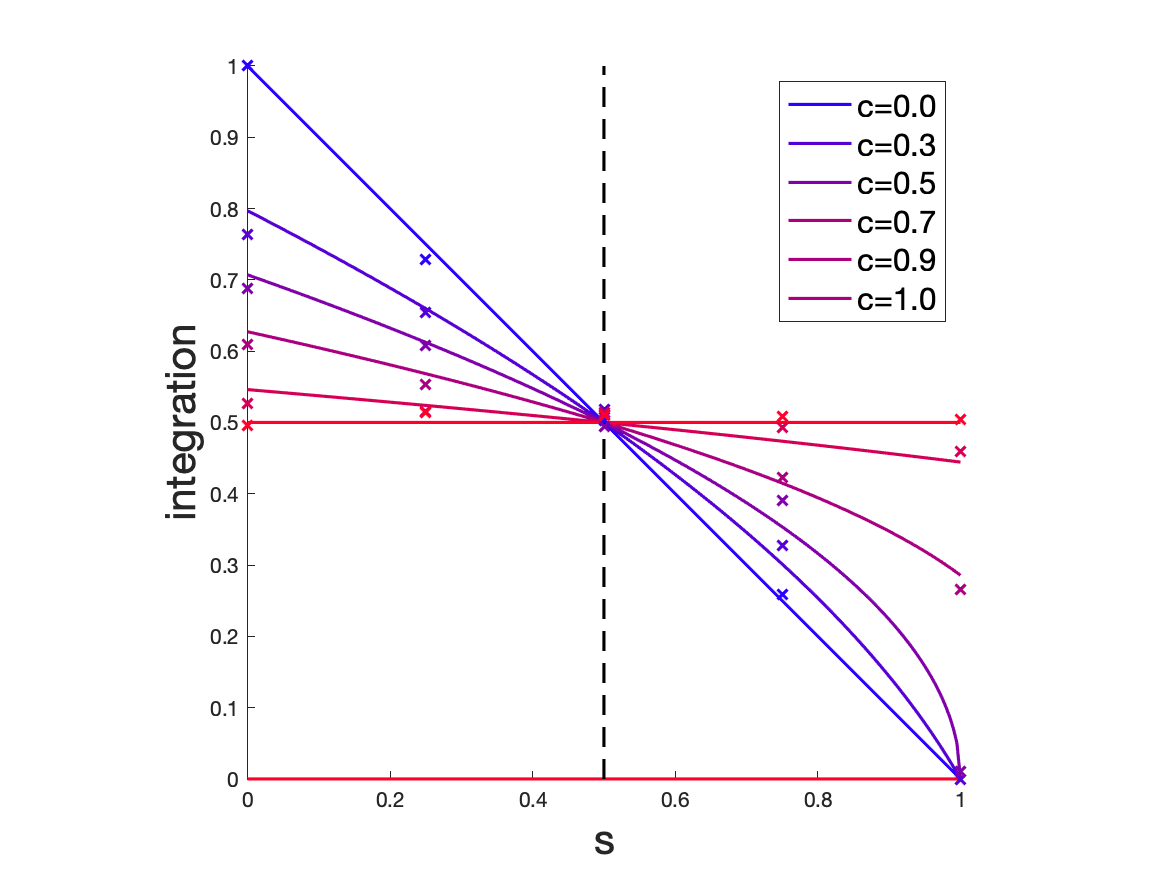}
    \caption{Network integration obtained theoretically from the fixed-node evolving model with equiprobable types and $s'=\frac{1}{2}$. Simulation results are also marked with crosses.}
    \Description{Network integration is plotted vs $s$ for a fixed-node evolving model with equiprobable types and $s'=\frac{1}{2}$. In the presence of homophily, triadic closure improves network integration.}
    \label{fig:asikainen_integ_vs_s}
\end{figure}
\section{Experiments}
\label{sec:exp}

Our results so far focus on theoretical observations for the expected behavior of network properties under some approximations, for a large number of nodes, and in the limit of $t \rightarrow \infty$. In this section, we examine the validity of our results both through analysis of real data and simulations. Here we discuss the applicability of Theorem~\ref{thm:identity_1} on real data and present simulation-based evidence in Section~\ref{sec:sim} of the appendix.

\subsection{Data: Citation Networks}

The citation network we study here is known to be captured well by the Jackson-Rogers model \cite{bramoulle2012homophily,jr}. Several factors make the citation network consistent with this model: First, papers---which correspond to nodes on this graph---appear sequentially and do not disappear. Likewise, citations---which correspond to edges on this graph---are directed and do not disappear over time. Third, researchers often use an initial seed of articles as a foundation for their work and use the citation network to identify further related works, similar to the second phase of the Jackson-Rogers model. 

We use the network of citations extracted mainly from DBLP, ACM, and MAG (Microsoft Academic Graph)~\cite{citation_network} (version 12). In this dataset, each paper is labeled with weighted fields of study. We use the field of study with the highest weight as the node type. The original dataset covers papers published mainly from 1960 to 2020. However, areas of study and access to articles have experienced a tremendous change during the last decades. We therefore focus on a shorter period of 2015 to 2020 to ensure network parameters are not varying over time. This period consists of more than 1.5 million articles with more than 4.7 million intra-citations, i.e., citations within the studied network. 

\subsection{Estimation of Model Parameters}
The citation network consists of papers from various fields. We limit our analysis to major fields of study, which we define to be fields that appear in at least $1$ percent of articles. Only $3$ percent of papers are not related to any major field. 

Despite the growth of interdisciplinary works, different fields of study still follow different publication traditions, resulting in different model parameters. We therefore first cluster the fields of study and then fit a separate model on each cluster, neglecting inter-cluster citations. Variation across clusters also provides further ability to test the validity of our theoretical findings. 

\vspace{2mm}
\noindent \textbf{Clustering Fields of Study.} In order to cluster fields, we obtain a weighted graph over major fields: Nodes correspond to fields in this graph and the weight of edge~$(f_1, f_2)$ corresponds to how many times a paper in field~$f_1$ has cited another paper with field~$f_2$ or vice versa. Note that papers may have more than one field of study. We then use spectral clustering to obtain clusters of major fields. Here we selected the number of clusters to be $6$ based on the eigenvalues of the graph's Laplacian matrix. Table~\ref{tab:cluster_summary} shows some statistics of these clusters. We chose the names of the clusters, looking at their most frequent fields.

\begin{table}[t]
    \centering
    \begin{tabular}{c|c|c|c}
         Cluster & Num. of Fields & Num. of Papers & Num. of within-cluster citations   \\
         \hline
         Mathematics & 13 & 487,298 & 954,544 \\
         Artificial Intelligence & 22 & 1,170,199 & 3,586,018 \\
         Knowledge Management & 15 & 386,073 & 617,090 \\
         Electrical Engineering & 20 & 474,773 & 1,172,721 \\
         Software Engineering & 4 & 85,845 & 90,723 \\
         Control Engineering & 7 & 254,453 & 355,930
    \end{tabular}
    \caption{Summary statistics of the clusters of major fields reported.}
    \label{tab:cluster_summary}
\end{table}

\vspace{2mm} 
\noindent \textbf{Estimation of Each Cluster's Model Parameters.} Assuming the network evolves according to the Jackson-Rogers model, we want to estimate model parameters from data. These parameters include $N_S$, $N_D$, $N_F$, and $\alpha$. To account for the randomness of real data, we add extra randomness here: at each time step, when a new node $u$ arrives, it draws model parameters from
\begin{align}
    N_S^{(u)}&\sim\exp \left( {\frac{1}{n_s} } \right) \nonumber\\
    N_D^{(u)}&\sim\exp\left( { \frac{1}{n_d} } \right) \nonumber\\
    N_{F,S}^{(u)}=\alpha^{(u)} N_F^{(u)}&\sim\exp\left( { \frac{1}{n_{f,s}} } \right) \nonumber \\
    N_{F,D}^{(u)}=\left( { 1-\alpha^{(u)} } \right) N_F^{(u)}&\sim\exp\left( { \frac{1}{n_{f,d}} } \right)
    .
\end{align}
Here $\exp(\lambda)$ corresponds to an exponential distribution with mean $\frac{1}{\lambda}$. We chose exponential priors only for simplicity. We believe similar results can be obtained with other positive distributions as well. Our goal is to estimate $\theta=(n_s, n_d, n_{f, s}, n_{f,d})$. 

In order to estimate $\theta$, we need to distinguish edges created during phase one and phase two. Let $G(u)=(\gV(u), \gE(u))$ be the induced subgraph of the citation graph over node (paper) $u$'s immediate descendants. Note, $G(u)$ does not include $u$. We want to know among all the edges from $u$ to $\gV(u)$ which ones are formed in the first and second phases. There is no way to distinguish first and second-phase connections. \citet{bramoulle2012homophily} suggest that if $(v, w) \in \gE(u)$, then $(u, v)$ is formed initially and $(u, w)$ is formed in the second phase due to triadic closure. However, we believe this assumption will add a bias to our estimation from model parameters. It is also not clear how to decide when there is a third node $x$ such that $(w, x) \in \gE(u)$ (Figure~\ref{fig:induced_graph_ex}). We propose an approach to estimate model parameters with minimum assumptions in the following.

Let $\phi_u: \gV(u) \rightarrow \{1, 2\}$ be the \textit{phase assignment function} for node $u$. $\phi_u(v)$ determines whether $(u, v)$ is created at the first or second phase. We call a phase assignment function \textit{feasible} if for every $w \in \gV(u)$ such that $\phi_u(w)=2$, there exists a $v \in \gV(u)$ such that $(v, w) \in \gE(u)$ and $\phi_u(v)=1$. In other words, if $(u, w)$ is assigned to be shaped in the second phase, there should be at least one mediator node that $u$ could find $w$ through it. Figure~\ref{fig:feasible_assignments_ex} shows an example of all feasible assignments of a graph with four nodes. 

\begin{figure}[t]
\centering
\begin{minipage}{0.24\linewidth}
    \centering
    \includegraphics[height=2cm]{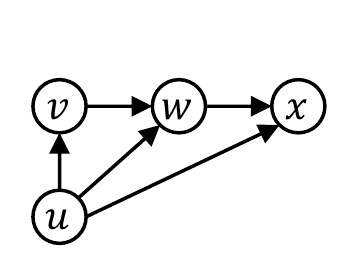}
    \caption{An example graph consists of $u$ and $G(u)$.}
    \Description{An example graph consists of edges~$(u,v)$, $(u,w)$, $(u,x)$, $(v,w)$, and $(w,x)$.}
    \label{fig:induced_graph_ex}
\end{minipage}
\hfill
\begin{minipage}{0.74\linewidth}
    \centering
    \includegraphics[height=2cm]{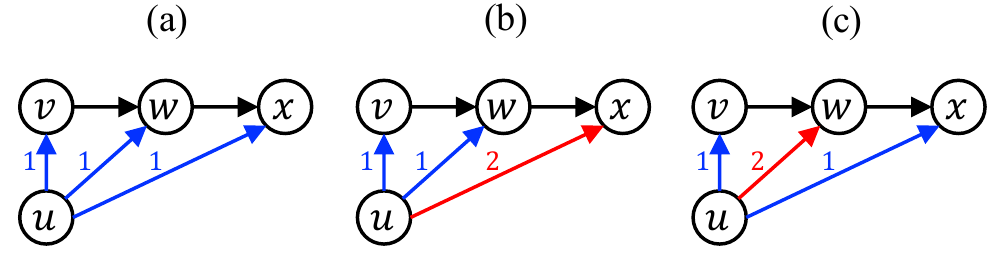}
    \caption{All feasible assignments for the edges from $u$ to $G(u)$. The numbers on the edges indicate the phase of edge formation.}
    \Description{Three possible assignments for the example graph are shown.}
    \label{fig:feasible_assignments_ex}
\end{minipage}
\end{figure}

Given an assignment function and a $type(\cdot)$ function, it is straightforward to find first phase parameters for node $u$:
\begin{align}
    n_s^{(u)}(\phi_u) &= |\{v \in \gV(u) \mid \phi_u(v)=1, type(u)=type(v)\}| \\
    n_d^{(u)}(\phi_u) &= |\{v \in \gV(u)  \mid \phi_u(v)=1, type(u)\neq type(v)\}|
    .
\end{align}
However, suppose an edge like $(u, w)$ is formed in phase two, and $w$ has immediate ancestors in $G(u)$ with both similar and dissimilar types to $u$. In that case, it is not clear whether $u$ and $w$ are connected through similar initial friends of $u$ or dissimilar initial friends. Here we look at the ratio of $w$'s immediate ancestors which are similar to $u$ and use this number as an estimate for $n_{f,s}^{(u)}$. 
\begin{align}
    n_{f,s}^{(u)}(\phi_u) &= \sum_{w: \phi_u(w)=2} 
    \frac{|\{v \in \gV(u) \mid (v, w)\in\gE(u), \phi_u(v)=1, type(v)=type(u)\}|}
    {|\{v \in \gV(u) \mid (v, w)\in\gE(u), \phi_u(v)=1\}|} 
\end{align}

We can do the same for ancestors which are dissimilar to $u$ to estimate $n_{f,d}^{(u)}$:

\begin{align}
 n_{f,d}^{(u)}(\phi_u) &= \sum_{w: \phi_u(w)=2} 
    \frac{|\{v \in \gV(u) \mid (v, w)\in\gE(u), \phi_u(v)=1, type(v)\neq type(u)|\}}
    {|\{v \in \gV(u) \mid (v, w)\in\gE(u), \phi_u(v)=1\}|}
   \end{align}

Let $\Phi_u$ denote the set of all feasible assignments for node~$u$. We assume a uniform distribution over $\Phi_u$. We can now find the likelihood of observing $G(u)$ given $\theta=(n_s, n_d, n_{f,s}, n_{f,d})$:
\begin{equation}
    l_u(\theta) = \frac{1}{|\Phi_u|} \sum_{\phi_u \in \Phi_u}
    \frac{1}{n_s\, n_d\, n_{f,s}\, n_{f, d}}
    \exp{\Big(-\frac{n_s^{(u)}(\phi_u)}{n_s}-\frac{n_d^{(u)}(\phi_u)}{n_d}-\frac{n_{f,s}^{(u)}(\phi_u)}{n_{f,s}}-\frac{n_{f,d}^{(u)}(\phi_u)}{n_{f,d}}\Big)}
\end{equation}

Finally, we maximize the likelihood of observing the whole cluster as it is to find the optimum parameters:
\begin{align}
    \theta^* = \argmax_{\theta} \sum_u \log{(l_u(\theta))}  
    .
\end{align}

We use the BFGS algorithm for the optimization; note, however, that this is a non-convex problem, and there is no guarantee that we can find the global maximum.

\begin{figure}[ht!]
\centering
\begin{minipage}{0.45\linewidth}
    \centering
    \includegraphics[width=1\linewidth]{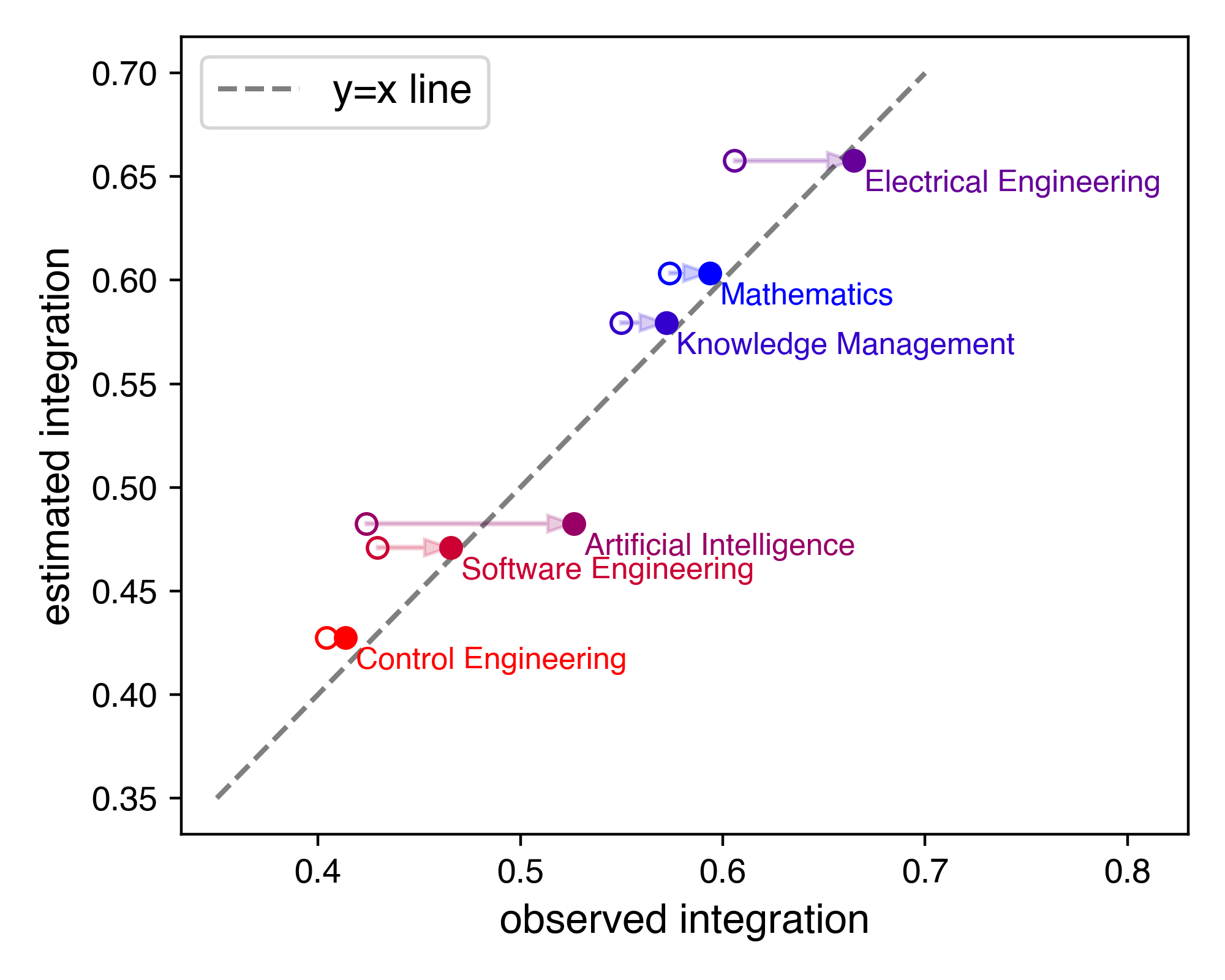}
    \caption{Estimated vs. observed integration. Theorem~\ref{thm:identity_1} is used to estimate network integration for different clusters. Empty and filled circles correspond to beginning and end year of the study, respectively.}
    \Description{Estimated and observed integration are very close except for one field.}
    \label{fig:observed_and_estimated_integ}
\end{minipage}
\hfill
\begin{minipage}{0.45\linewidth}
    \centering
    \includegraphics[width=1\linewidth]{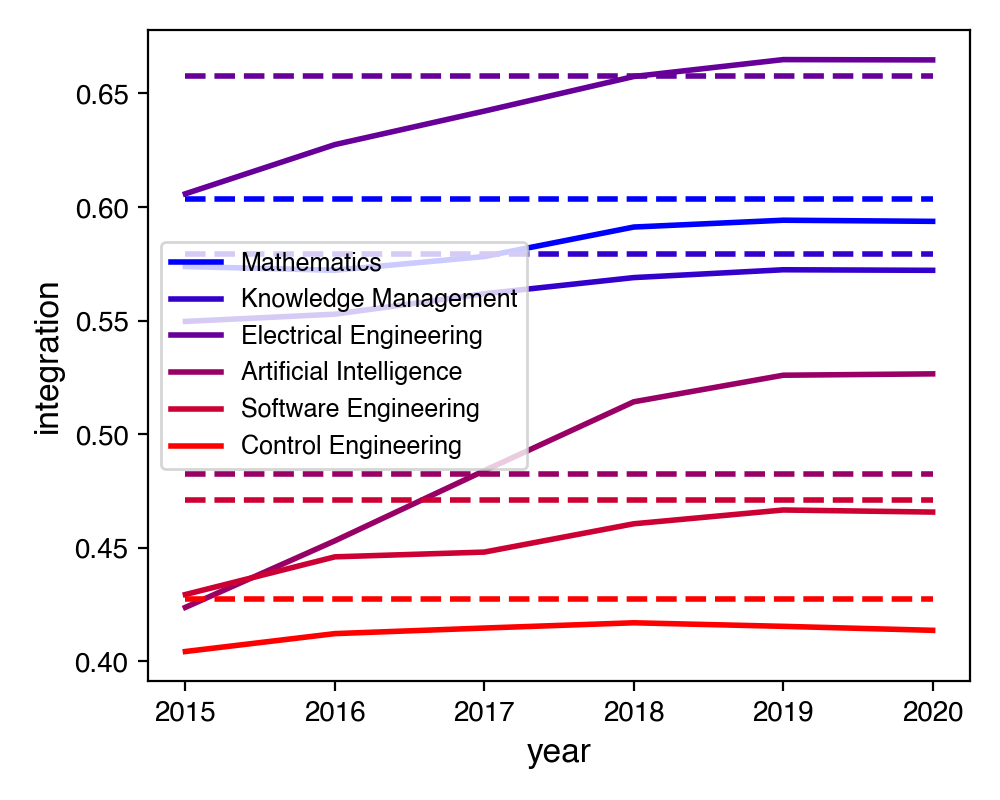}
    \caption{Convergence behavior of observed integration (solid) to estimated integration (dashed). With the exception of one field (Artificial Intelligence), network integration in other fields has converged to the predicted value.}
    \Description{Observed integration converges to the estimated integration by time except for one field.}
    \label{fig:observed_and_estimated_integ_in_time}
\end{minipage}
\end{figure}


\subsection{Results of Citation Network Analysis}
We use the obtained optimum parameters $\theta^*$ and Theorem~\ref{thm:identity_1} to estimate network integration in equilibrium. Figure~\ref{fig:observed_and_estimated_integ} shows the estimated integration in equilibrium versus the observed integration. In this figure, empty and filled marks correspond to the starting year (2015) and final year (2020) respectively. With the exception of one cluster (Artificial Intelligence), different clusters consistently approach our estimated values from equilibrium. The convergence behavior is also depicted in Figure~\ref{fig:observed_and_estimated_integ_in_time}. These empirical insights show that even with the assumptions needed for Theorem~\ref{thm:identity_1}, the theoretical insights closely match practice in this dataset. As we see from Figure~\ref{fig:observed_and_estimated_integ}, although clusters have various fields with different frequencies, their behavior in equilibrium is well-predicted from the theory with only a few parameters. Our empirical findings present further evidence that the Jackson-Rogers model explains citation network evolution. Finally, with estimated parameters, we find triadic closure to be responsible for $3$-$5\%$ of network integration. We do so by setting $N_F=0$ in Theorem~\ref{thm:identity_1}, as a proxy for network integration without triadic closure. 
\section{Further Related Works} 
\label{sec:related_works}

Homophily is a robust and prevalent process impacting network formation in many domains~\cite{lazarsfeld1954friendship,mcpherson2001birds,newman2002assortative}. There is a long line of theoretical and empirical work exploring the effect of homophily on network formation, ranging from observational studies on large network data, to laboratory experiments, to analyses of theoretical models \cite{ag,dong2017structural,gru}. 

A main topic of focus has been the interaction between homophily and network segregation. For instance, \citet{currarini2009economic,pnasschelling} show that segregated networks emerge due to homophily. In related work, \citet{ka} study the effect of homophily on the rich-get-richer phenomena. Empirical work has explored the effect of homophily on integration in settings like college campuses \cite{mp}. In related work to ours, \citet{bramoulle2012homophily} adapt the Jackson-Rogers model to the case with heterogeneous nodes. Their work primarily focuses on how each node's likelihood to form links changes over time. In contrast, we consider a global measurement of integration, using the fraction of bichromatic edges.


Triadic closure is another well-studied process in network formation dating back to the 1950s~\cite{kossinets2006empirical,rapoport1953spread}. While there is a long line of work on the effect of triadic closure on network clustering, the interplay between homophily and triadic closure remains under-explored. In one complementary related work, \citet{altenburger2018monophily} show that monophily---the presence of individuals with preference for attributes unrelated to their own---has a tendency to induce similarity among friends-of-friends. In contrast, we study the relationship of homophily and triadic closure, though some of the findings complement our observations. 

The closest work to ours is that by \citet{asikainen}, which explores the interaction of triadic closure and homophily. Here, we are similarly concerned with how these two phenomena interact in dynamic models. \citet{asikainen} consider a model that also starts with an SBM and adds both triadic closure and random link rewiring. Both of these additions are influenced by choice homophily. Under this model, \citet{asikainen} show that triadic closure amplifies the effects of homophily. We note, however, the model considered here already has the triadic closure step influenced by homophily. In our work, we make minimal adjustments to their model to further isolate the effects of triadic closure and find results consistent with the SBM and Jackson-Rogers models. 

Segregation in social networks can limit individuals' ability to access information, resources, and opportunities, leading to the creation or exacerbation of disparities across groups. Research across various disciplines has modeled and measured the impact of segregation on social welfare including its impacts on access to information, economic development, educational outcomes, labor market outcomes, and social capital and support~\cite{banerjee2013diffusion,CalvoJackson,ccz,krishna,dimaggio2011network,eagle2010network,jackson2012social}
Recent work, such as by \citet{avin2015glass} has proposed and studied models that explain how inequality and disparities in access to opportunity arise in certain settings. 

Our work has additional implications for network-based interventions both in on- and off-line settings. For instance, it is well-known that
biases that may exist on online platforms such as Twitter and Task Rabbit may lead to inequalities between groups~\cite{twitter,taskrabbit}. These biases, amplified by recommendation algorithms, can impact how networks grow and evolve creating an \emph{algorithmic glass ceiling}~\cite{biega2018equity,stoica2018algorithmic,su2016effect,biega2018equity}. In recent years, there has been interest by researchers in algorithmically-informed interventions that can help better diagnose and mitigate underlying patterns of inequality on platforms~\cite{abebe2018mechanism,abebe2020roles}. Focused on fairness in recommender systems, researchers have examined the effect of small interventions on the long-term health of the platform such as by mitigating segregation, improving interactions, and achieving other desirable societal objectives~\cite{ekstrand2016behaviorism,guy2015social,hutson2018debiasing,knijnenburg2016recommender,schnabel2018short,stoica2018algorithmic,su2016effect}. These studies have shown that the platform designer, by using small interventions when a user first joins, may be able to realize large gains on the platform health over time. 

\section{Discussion and Conclusion}
\label{sec:discussion}

In this work, we consider the effect of triadic closure on network segregation. Through analyses of different static and dynamic network formation models, we find that triadic closure has the effect of increasing network integration, indicating that it may be a process that counteracts homophily in network formation. 

We find it striking that such a tension should exist between two such well-studied social processes as homophily and triadic closure. In addition to the theoretical and empirical results tackled in this work, we believe this counter-intuitive result about the relationship between homophily and triadic closure points to a rich and under-explored phenomenon about their interaction. 

These results also open up questions related to other measurements of network health, such as network expansion and distribution of network centralities. Each of these points to challenging analytic questions. Empirically, it would also be interesting to shed light on what types of social and information networks tend to exhibit a stronger relationship between triadic closure and homophily.

Finally, the interventions presented in this work point to a broader set of theoretical and empirical questions. For instance, it would be interesting to estimate the various network parameters and compare the effect of nudges across different distributions of values. Furthermore, such interventions are often costly to the designer
or may incur social cost, leading to a set of optimization questions where the designer must trade off these costs
with utility gained from network integration.

\begin{acks}
We thank Sera Linardi, Emma Forman Ling, Irene Lo, Ashudeep Singh, Ana-Andreea Stoica, Sam Taggart, Bryan Wilder, Angela Zhou, and members of the MD4SG Working Group on Inequality for helpful discussions throughout the evolution of this work. We especially thank Emma Forman Ling for numerous discussions and pointers to the citation network. We additionally thank the reviewers, area chairs, track chairs, and program chairs of EC '22 for their insightful feedback. 
\end{acks}

\bibliographystyle{template/ACM-Reference-Format}
\bibliography{refs}

\newpage
\appendix
\section{Additional Statements and Proofs}
\label{sec:missing_proofs}

\begin{lemma}
\label{lemma:n_edges_sbm}
For any $\text{SBM}(p, q)$~network~$G$ with $K \ge 2$~types each consisting of $n_k$~nodes ($k \in [K]$), the expected number of monochromatic edges is
\begin{align}
    e_m = \sum_{k \in [K]} {n_k \choose 2} p \nonumber = \frac{1}{2} p \sum_k n_k^2 + O(n)
    ,
\end{align}
and the expected number of bichromatic edges is
\begin{align}
    e_b = \frac{1}{2} \sum_{k \in [K]} \sum_{l \in [K], l \neq k} n_k n_l q = \frac{1}{2} q (n^2 - \sum_k n_k^2) + O(n)
    ,
\end{align}
where $n = \sum_k n_k$. Further, the expected number of monochromatic \emph{missing} edges is
\begin{align}
    o_m = \sum_{k \in [K]} {n_k \choose 2} (1 - p) \nonumber = \frac{1}{2} (1 - p) \sum_k n_k^2 + O(n)
    ,
\end{align}
and the expected number of bichromatic missing edges is
\begin{align}
    o_b = \frac{1}{2} \sum_{k \in [K]} \sum_{l \in [K], l \neq k} n_k n_l (1 - q) = \frac{1}{2} (1 - q) \big(n^2 - \sum_k n_k^2 \big) + O(n)
    ,
\end{align}
where we call an edge~$(i,j)$ a missing edge if $i$ and $j$ are not connected in $G$.
\end{lemma}
\begin{proof}
To find $e_m$ ($o_m$), we sum the number of unordered pairs from each group times the probability they are connected (not connected). To find $e_b$ ($o_b$), we count the number of bichromatic pairs times the probability they are connected (not connected) times $\frac{1}{2}$ to compensate for repeated counting. 
\end{proof}

\begin{lemma}
\label{lemma:n_wedges_sbm}
For the same network as Lemma \ref{lemma:n_edges_sbm}, the expected number of monochromatic wedges is
\begin{align}
    w_m &= \sum_{k \in [K]} {n_k \choose 2} (n_k - 2) p^2 (1 - p) 
    + \sum_{k \in [K]} {n_k \choose 2} (n - n_k) q^2 (1 - p) \nonumber \\
    &= \frac{1}{2} p^2 (1 - p) \sum_k n_k^3 + 
    \frac{1}{2} q^2 (1 - p)\big(n \sum_k n_k^2 - \sum_k n_k^3 \big) + O(n^2)
    ,
\end{align}
and the expected number of bichromatic wedges is
\begin{align}
    w_b &= \sum_{k \in [K]} \sum_{l \in [K], l \neq k} n_k (n_k - 1) n_l p q (1 - q) + 
    \sum_{k \in [K]} \sum_{l \in [K], l \neq k} \frac{1}{2} n_k n_l (n - n_k - n_l) q^2 (1 - q) \nonumber \\
    &= p q (1 - q) \big( n \sum_k n_k^2 - \sum_k n_k^3 \big)
    + \frac{1}{2} q^2 (1 - q) \big(n^3 + 2 \sum_k n_k^3 - 3n \sum_k n_k^2 \big) + O(n^2)
\end{align}
\end{lemma}
\begin{proof}
For a wedge $i\text{---}j\text{---}k$, the first term of $w_m$ is the expected number of wedges such that $type(i)=type(j)=type(k)$. The second term of $w_m$ corresponds to the case $type(i)=type(k)\neq type(j)$. The first term of $w_b$ is the expected number of wedges such that $type(i)=type(j)\neq type(k)$. The second term of $w_b$ corresponds to the case where $i$, $j$, and $k$ are all from different types.
\end{proof}

\begin{lemma}
\label{lemma:n_m_ineq}
For any set of $\{n_1, n_2, \cdots, n_K | n_i \in \R^+\}$, following inequalities hold:
\begin{enumerate}
    \item $n^2 \ge m_2$
    \item $n m_2 \ge m_3$
    \item $n m_3 \ge m_2^2$
    \item $2 m_2^2 \ge n m_3$
    \item $2 n m_2^2 \ge n^2 m_3 + m_2 m_3$,
\end{enumerate}
where $n = \sum_{k \in [K]} n_k$ and $m_i = \sum_{k \in [K]} n_k^i$.
\end{lemma}
\begin{proof}
We start by the first inequality and use the results to that point in proving each inequality. 

\begin{enumerate}
\item
\begin{align*}
    n^2 - m_2 &= \sum_{i, j} n_i n_j - \sum_j n_j^2 \\
    &= \sum_{i \neq j} n_i n_j \ge 0
    .
\end{align*}

\item
\begin{align*}
    n m_2 - m_3 &= \sum_{i, j} n_i n_j^2 - \sum_j n_j^3 \\
    &= \sum_{i \neq j} n_i n_j^2 \ge 0
    .
\end{align*}

\item
\begin{align*}
    n m_3 - m_2^2 &= \sum_{i, j} n_i n_j^3 - n_i^2 n_j^2 \\
    &= \sum_{i \neq j} n_i n_j^3 + n_j n_i^3 - 2 n_i^2 n_j^2 \\
    &= \sum_{i \neq j} n_i n_j (n_i - n_j)^2 \ge 0
    .
\end{align*}

\item
\begin{align*}
    2 m_2^2 - n m_3 &= \sum_{i, j} 2 n_i^2 n_j^2 - n_i n_j^3 \\
    &= \sum_{i \neq j} n_i^4 + n_j^4 + 4 n_i^2 n_j^2 - n_i n_j^3 - n_j n_i^3 \\
    &= \sum_{i \neq j} (n_i^2 + n_j^2)^2 - n_i n_j(n_i - n_j)^2 
    .
\end{align*}
As $n_i, n_j \ge 0$, $n_i^2 + n_j^2 \ge (\max(n_i, n_j))^2 \ge n_i n_j$ and $n_i^2 + n_j^2 \ge (\max(n_i, n_j))^2 \ge (n_i - n_j)^2$. So, $(n_i^2 + n_j^2)^2 \ge n_i n_j(n_i - n_j)^2$, which gives $2 m_2^2 - n m_3 \ge 0$.

\item 
\begin{align*}
    2 n m_2^2 - n^2 m_3 - m_2 m_3 &= (\sum_{i, j} 2 n_i^2 n_j^2 - n_i n_j^2) n - \sum_{i, j} n_i^2 n_j^3 \\
    &= \sum_{i \neq j} (n_i^4 + n_j^4 + 4 n_i^2 n_j^2 - n_i n_j^3 - n_j n_i^3) n - \sum_{i \neq j} n_i^5 + n_j^5 + n_i^2 n_j^3 + n_j^2 n_i^3 \\
    &= \sum_{i \neq j} ((n_i^2 + n_j^2)^2 - n_i n_j(n_i - n_j)^2 )(n - n_i - n_j) \\
    &+ \sum_{i \neq j} (n_i^4 + n_j^4 + 4 n_i^2 n_j^2 - n_i n_j^3 - n_j n_i^3)(n_i + n_j) - n_i^5 + n_j^5 + n_i^2 n_j^3 + n_j^2 n_i^3 \\
    &\ge \sum_{i \neq j} (n_i^4 + n_j^4 + 4 n_i^2 n_j^2 - n_i n_j^3 - n_j n_i^3)(n_i + n_j) - (n_i^5 + n_j^5 + n_i^2 n_j^3 + n_j^2 n_i^3) \\
    &= \sum_{i \neq k} 2 n_i^2 n_j^3 + 2 n_j^2 n_i^2 \ge 0
    .
\end{align*}
Here we used $(n_i^2 + n_j^2)^2 - n_i n_j(n_i - n_j)^2 \ge 0$ from the proof of the previous part.

\end{enumerate}
\end{proof}

\begin{lemma}
\label{lemma:inv_p}
Let $\mA$ be a $K \times K$ real symmetric matrix such that
\begin{equation}
    [\mA]_{i, j} = \begin{cases}
      c & i = j \\
      \frac{1}{K-1} (1 - c) & o.w.
    \end{cases}
    .
\end{equation}
Then, $\mA$ has $K$ real eignvalues:
\begin{equation}
    d_i = \begin{cases}
      1 & i = 1 \\
      \frac{Kc - 1}{K - 1} & 1 < i \le K
    \end{cases}
    ,
\end{equation}
with corresponding eigenvectors:
\begin{align}
    [\vv_1]_j &= 1 \nonumber \\
    \text{For $i \ge 2$: } [\vv_i]_j &= \begin{cases}
      1 & j = 1 \\
      -1 & j = i \\
      0 & o.w.
    \end{cases}
    .
\end{align}
\end{lemma}

\begin{proof} For $i = 1$:
\begin{equation}
    [\mA \vv_1]_j = 1 = d_1 [\vv_1]_j
    .
\end{equation}

For $1 < i \le K$:
\begin{equation}
    [\mA \vv_i]_j = \begin{cases}
      \frac{Kc - 1}{K - 1} & j = 1 \\
      -\frac{Kc - 1}{K - 1} & j = i \\
      0 & o.w.
    \end{cases}
    = d_i [\vv_1]_j
    .
\end{equation}
\end{proof}

\begin{lemma}
\label{lemma:inv_V}
The inverse of the $K \times K$ matrix $\mV$ defined by
\begin{equation}
    [\mV]_{i, j} = \begin{cases}
      1 & i = 1 \text{ or } j = 1 \\
      -1 & i = j > 1 \\
      0 & o.w.
    \end{cases}
\end{equation}
is 
\begin{equation}
    [\mV^{-1}]_{i, j} = \begin{cases}
      -\frac{K-1}{K} & i = j > 1 \\
      \frac{1}{K} & o.w.
    \end{cases}
    .
\end{equation}
\end{lemma}
\begin{proof}
\begin{equation}
    [\mV \mV^{-1}]_{i, j} = \begin{cases}
      1 & i = j = 1 \\
      \frac{1}{K} + \frac{K - 1}{K} = 1 & i = j > 1 \\
      0 & o.w. 
    \end{cases}
    .
\end{equation}
\end{proof}

\printProofs

\section{Simulations}
\label{sec:sim}

First of all, we investigate stochastic block models through simulation. Specifically, we test Theorem~\ref{thm:sbm:gamma_rel_eff_on_integ} as Theorems~\ref{thm:sbm:abs_eff_on_integ}~and~\ref{thm:sbm:rel_eff_on_integ} can be seen as special cases of this theorem for $\gamma=0$ and $\gamma \rightarrow \infty$, respectively. From Theorem~\ref{thm:sbm:gamma_rel_eff_on_integ} we expect triadic closure to have positive relative effect when $u(\gamma) > \frac{p}{q} > l(\gamma)$. To test this theorem, we have simulated a stochastic block model consisting of $K$ groups, where group $k$ has $n_k=10 \, \lambda^k$ members. We use various values of $\frac{p}{q}$ and $\lambda$ in simulations. Figures~\ref{fig:sim_vs_theory(lambda1_k5)}~and~\ref{fig:sim_vs_theory(lambda2_k2)} show the theoretical and simulated results together. Each mark on the figures shows specific values of $\frac{p}{q}$ and $\gamma$ that triadic closure has had a positive relative effect on average. One can see despite having small networks, Theorem~\ref{thm:sbm:gamma_rel_eff_on_integ} well predicts the effect both for balanced and unbalanced networks.

\begin{figure}[t]
\begin{minipage}{0.49\linewidth}
    \centering
    \includegraphics[width=0.95\linewidth]{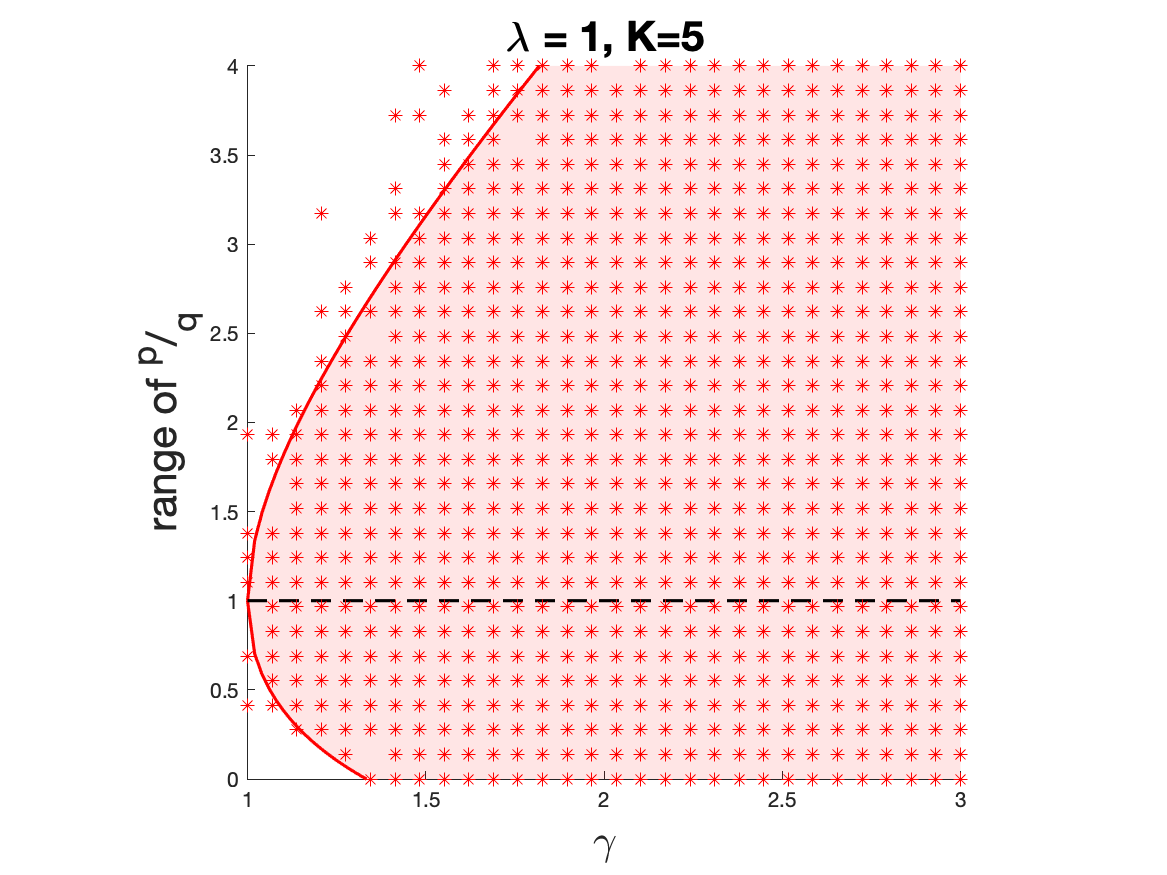}
    \caption{$u(\gamma)$ and $l(\gamma)$ for balanced groups. The shaded area is obtained by theory. Marks show positive relative effects in simulations.}
    \label{fig:sim_vs_theory(lambda1_k5)}
\end{minipage}
\hfill
\begin{minipage}{0.49\linewidth}
    \centering
    \includegraphics[width=0.95\linewidth]{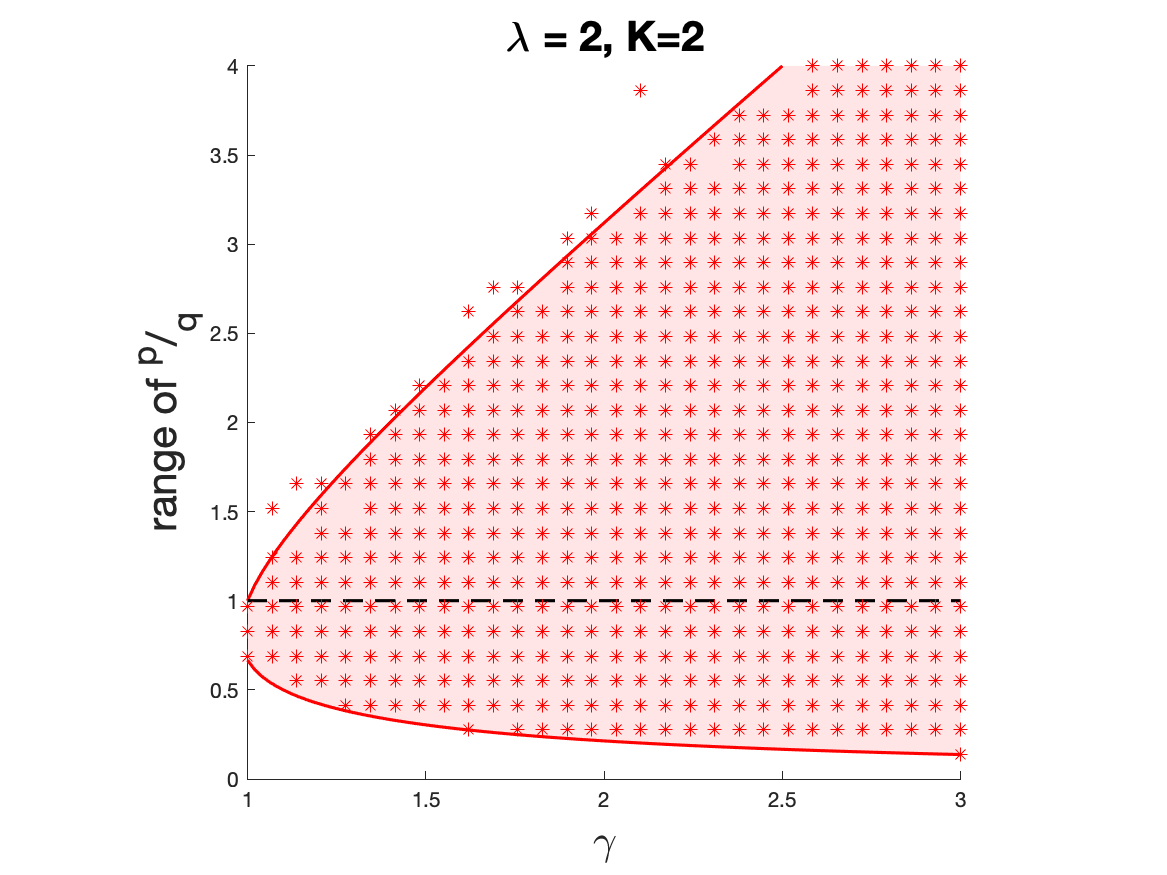}
    \caption{$u(\gamma)$ and $l(\gamma)$ for unbalanced groups. The shaded area is obtained by theory. Marks show positive relative effects in simulations.}
    \label{fig:sim_vs_theory(lambda2_k2)}
\end{minipage}
\end{figure}

Next, we study the Jackson-Rogers model's convergence through simulations. Figure~\ref{fig:jr_convergence(lin)} shows network integration of a dynamic graph consisting of two groups evolving with the Jackson-Rogers model. The dashed lines show the integration in equilibrium predicted by Theorem~\ref{thm:identity_1}. The solid and dotted lines show the average integration of repeated simulations for balanced and unbalanced networks. To show that behavior in equilibrium is independent of the initial network, we have run simulations for two cases: a completely segregated initial network and a fully connected initial network. It seems Theorem~\ref{thm:identity_1} can robustly predict the network's behavior for various model parameters, regardless of the initial network and distribution of groups. Further, Figure~\ref{fig:jr_convergence(log)} shows the residual to equilibrium on a logarithmic scale. The dashed lines correspond to the convergence rate $O(t^{-\frac{N_S + N_D}{N}})$. One can see the proposed upper bound on the convergence rate also matches the behavior of the network in simulations.

\begin{figure}[t]
\begin{minipage}{0.49\linewidth}
    \centering
    \includegraphics[width=0.95\linewidth]{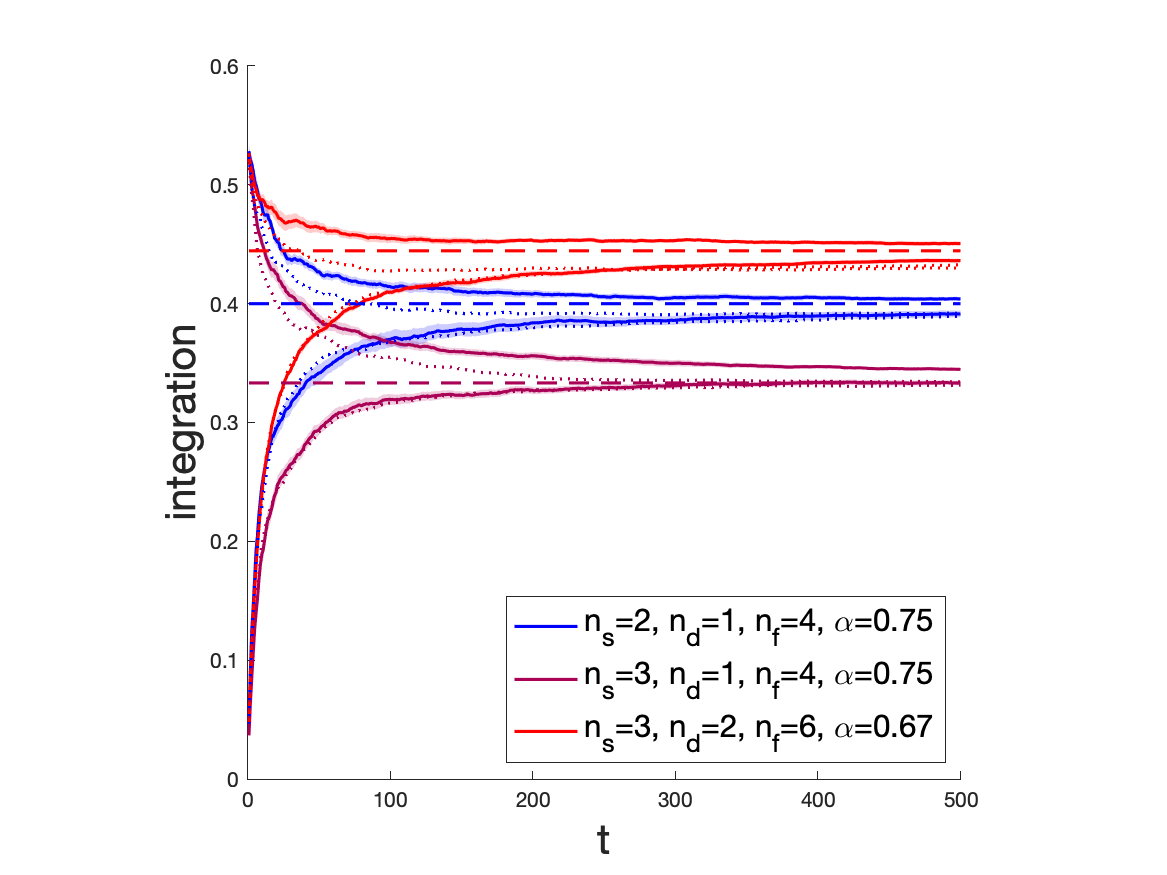}
    \caption{Convergence of a Jackson-Rogers network. Dashed lines show predicted behavior in equilibrium by Theorem~\ref{thm:identity_1}.}
    \label{fig:jr_convergence(lin)}
\end{minipage}
\hfill
\begin{minipage}{0.49\linewidth}
    \centering
    \includegraphics[width=0.95\linewidth]{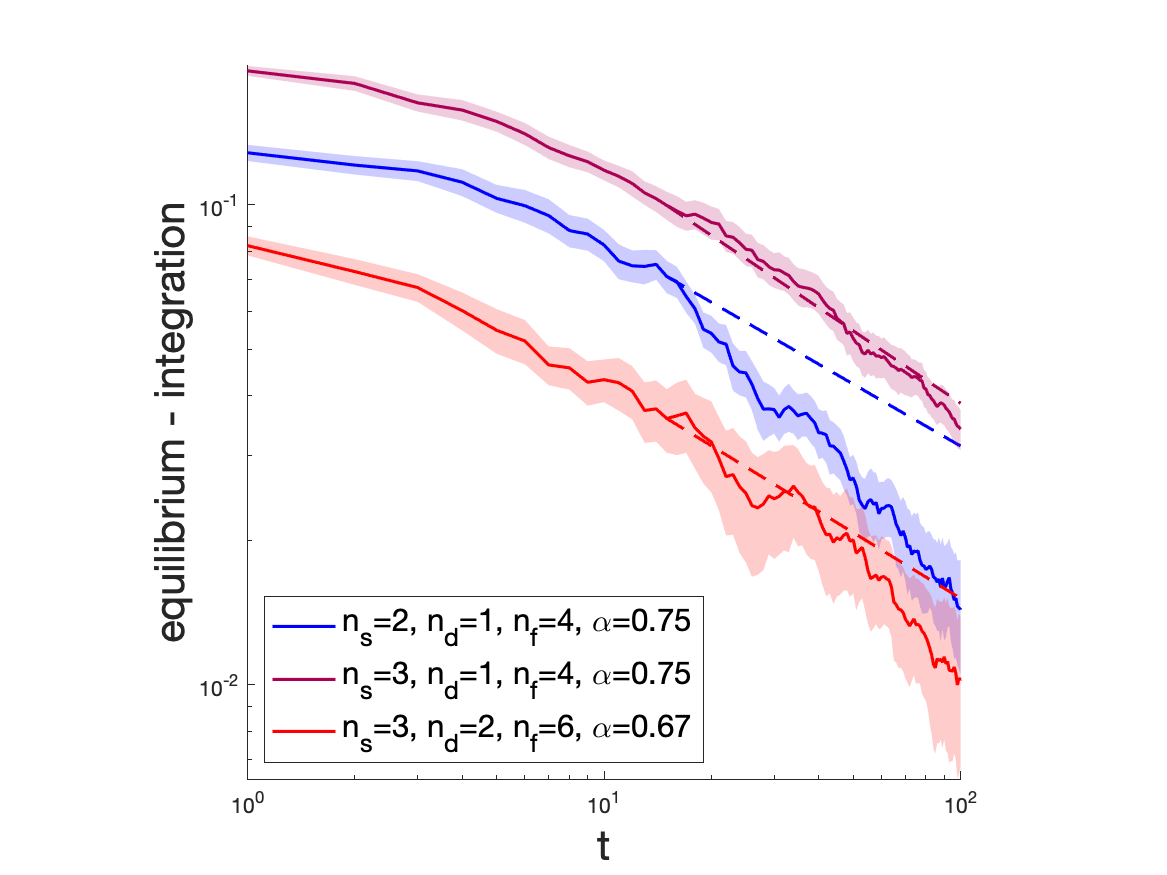}
    \caption{The residual to reach the predicted equilibrium. Dashed lines show predicted upper bound on the convergence rate.}
    \label{fig:jr_convergence(log)}
\end{minipage}
\end{figure}

\end{document}